\documentclass[11pt]{article}%
\usepackage{amsfonts}
\usepackage{amsmath}
\usepackage{amssymb}
\usepackage{graphicx}%
\providecommand{\U}[1]{\protect\rule{.1in}{.1in}}
\newtheorem{theorem}{Theorem}
\newtheorem{acknowledgement}[theorem]{Acknowledgement}

\newtheorem{corollary}[theorem]{Corollary}

\newtheorem{lemma}[theorem]{Lemma}

\newtheorem{proposition}[theorem]{Proposition}

\newenvironment{proof}[1][Proof]{\noindent\textbf{#1.} }{\ \rule{0.5em}{0.5em}}
\begin{document}

\title{ Symplectic Coarse-Grained Dynamics: Chalkboard Motion in Classical and
Quantum Mechanics}
\author{Maurice A. de Gosson\thanks{maurice.de.gosson@univie.ac.at}\\University of Vienna,\\Faculty of Mathematics (NuHAG)\\Oskar-Morgenstern-Platz 1, Vienna}
\maketitle
\tableofcontents

\begin{abstract}
In the usual approaches to mechanics (classical or quantum) the primary object
of interest is the Hamiltonian, from which one tries to deduce the solutions
of the equations of motion (Hamilton or Schr\"{o}dinger). In the present work
we reverse this paradigm and view the motions themselves as being the primary
objects. This is made possible by studying arbitrary phase space motions, not
of points, but of (small) ellipsoids with the requirement that the symplectic
capacity of these ellipsoids is preserved. This allows us to guide and control
these motions as we like. In the classical case these ellipsoids correspond to
a symplectic coarse graining of phase space, and in the quantum case they
correspond to the \textquotedblleft quantum blobs\textquotedblright\ we
defined in previous work, and which can be viewed as minimum uncertainty phase
space cells which are in a one-to-one correspondence with Gaussian pure states.

\end{abstract}

\section{Introduction}

In traditional classical and quantum mechanics it is assumed that the
Hamiltonian function (or its quantization) is given and one thereafter sets
out to solve the corresponding dynamical equations (Hamilton or
Schr\"{o}dinger). In the present paper we reverse this paradigm by considering
the primary objects as being motions, classical or quantum. These motions are
not defined by their actions on points, but rather on ellipsoids with constant
symplectic capacity, as motivated by our discussion above. We will see that
there is a great latitude in choosing these motions, justifying our use of the
metaphor \textquotedblleft chalkboard motion\textquotedblright: these motions
can be compared to chalk drawings on a blackboard leaving a continuous
succession of thick points. The surprising fact, which originally motivated
our study, is that chalkboard motions are indeed Hamiltonian, but this in a
very simple and unexpected way. We will be able to construct such motions at
will: exactly as when one stands in front of a blackboard and uses a piece of
chalk to make a drawing -- except that in our case the blackboard is infinite
and multidimensional, and the drawing consists of paths left by moving ellipsoids.

As we will see, these constructs allow us to define a quantum phase space,
obtained from the usual Euclidean phase\ space using a coarse-graining by
minimum uncertainty ellipsoids (we have dubbed these ellipsoids
\textquotedblleft quantum blobs\textquotedblright\ elsewhere \cite{blobs}). We
will then be able to define a \textquotedblleft chalkboard
motion\textquotedblright\ in this quantum phase space by quantizing the
classical chalkboard motions; this will again give us great latitude in
\textquotedblleft piloting\textquotedblright\ and controlling at each step
these quantum motions. As we will see, this procedure has many advantages, in
particular that of conceptual and computational simplicity. Admittedly, the
term \textquotedblleft quantum phase space\textquotedblright\ is usually
perceived as a red herring in physics: some physicists argue that there can't
be any phase space in quantum mechanics, since the notion of a well-defined
point does not make sense because of the uncertainty principle. Dirac himself
dismissed in 1945 in a letter to Moyal even the suggestion that quantum
mechanics can be expressed in terms of classical-valued phase space variables
(see Curtright \textit{et al}. \cite{curtright} for a detailed account of the
Dirac--Moyal discussion). Still, most theoretical physicists use phase space
techniques every day when they work with the Wigner functions of quantum
states: these functions are defined on the classical phase space
$\mathbb{R}_{x}^{n}\times\mathbb{R}_{p}^{n}$, and this does not lead to any
contradictions: the datum of the Wigner function $W\psi$ of a state $\psi$ is
both mathematically and physically equivalent to the datum of the state
itself. There are in truth many phase space approaches to quantum mechanics;
see for instance \cite{epi,polko,Olivares} for various and sometimes
conflicting points of view.

\subsection{Introductory example}

Let us consider the disk $D(\varepsilon):x^{2}+p^{2}\leq\varepsilon^{2}$ in
the phase plane $\mathbb{R}_{x}\times\mathbb{R}_{p}$. We smoothly deform this
disk into an ellipse; such a deformation is represented by a family of real
$2\times2$ matrices
\begin{equation}
S_{t}=%
\begin{pmatrix}
a_{t} & b_{t}\\
c_{t} & d_{t}%
\end{pmatrix}
\label{st}%
\end{equation}
and the disk thus becomes after time $t$ the ellipse $D_{t}(\varepsilon
)=S_{t}D(\varepsilon)$ represented by
\[
(c_{t}^{2}+d_{t}^{2})x^{2}+(a_{t}^{2}+b_{t}^{2})p^{2}-2(a_{t}c_{t}+b_{t}%
d_{t})px\leq\varepsilon^{2}.
\]
Assume now that the ellipses $S_{t}D(\varepsilon)$ all have the same area
$\pi\varepsilon^{2}$; the family $(S_{t})$ must then consist of symplectic
matrices,\textit{ i.e.} $\det S_{t}=a_{t}d_{t}-b_{t}c_{t}=1$. This constraint
implies that
\[
(a_{t}^{2}+b_{t}^{2})(c_{t}^{2}+d_{t}^{2})=1+\mu_{t}^{2}%
\]
where we have set $\mu_{t}=a_{t}c_{t}+b_{t}d_{t}$\ so we can rewrite the
equation of $D_{t}(\varepsilon)$ in the form%
\[
\frac{1+\mu_{t}^{2}}{a_{t}^{2}+b_{t}^{2}}x^{2}+(a_{t}^{2}+b_{t}^{2})p^{2}%
-2\mu_{t}px\leq\varepsilon^{2}%
\]
which shows that the ellipse $D_{t}(\varepsilon)$ can be obtained from the
disk $D(\varepsilon)$ using, instead of $S_{t}$, the family of lower
triangular matrices%
\begin{equation}
R_{t}=%
\begin{pmatrix}
\lambda_{t}^{-1} & 0\\
\mu_{t}\lambda_{t}^{-1} & \lambda_{t}%
\end{pmatrix}
\text{ \ \textit{with} }\left\{
\begin{array}
[c]{c}%
\lambda_{t}=1/\sqrt{a_{t}^{2}+b_{t}^{2}}\\
\mu_{t}=(a_{t}c_{t}+b_{t}d_{t})/(a_{t}^{2}+b_{t}^{2})
\end{array}
\right.  \label{reduced}%
\end{equation}
or, equivalently,
\begin{equation}
R_{t}=%
\begin{pmatrix}
1 & 0\\
\mu_{t} & 0
\end{pmatrix}%
\begin{pmatrix}
\lambda_{t}^{-1} & 0\\
0 & \lambda_{t}%
\end{pmatrix}
. \label{reducedbis}%
\end{equation}
This shows, in particular, that any ellipse can be obtained from a disk with
the same area and center using only a coordinate rescaling and a shear. All
this actually becomes much more obvious if one recalls that every symplectic
matrix can be factorized as a product of a shear, a rescaling, and a rotation:
this is called the \textquotedblleft Iwasawa factorization\textquotedblright,
which we will study in detail in Section \ref{seciwa}. This implies that there
is considerable redundancy when we let the $S_{t}$ act on circular disks
centered at the origin; in fact we could replace $S_{t}$ with any product
\[
S_{t}=%
\begin{pmatrix}
1 & 0\\
\mu_{t} & 0
\end{pmatrix}%
\begin{pmatrix}
\lambda_{t}^{-1} & 0\\
0 & \lambda_{t}%
\end{pmatrix}%
\begin{pmatrix}
\cos\alpha_{t} & \sin\alpha_{t}\\
-\sin\alpha_{t} & \cos\alpha_{t}%
\end{pmatrix}
\]
where $\alpha_{t}$ is a smoothly varying angle. We next note that the matrices
$S_{t}$ and $R_{t}$ being symplectic, the families $(S_{t})$ and $(R_{t})$ can
both be interpreted as Hamiltonian flows. These flows are generated by the
quadratic time-dependent Hamiltonian functions%
\begin{align}
H_{S}(x,p,t)  &  =\frac{1}{2}(a_{t}\dot{b}_{t}-\dot{a}_{t}b_{t})p^{2}+\frac
{1}{2}(\dot{d}_{t}c_{t}-\dot{c}_{t}d_{t})x^{2}-(a_{t}\dot{d}_{t}-b_{t}\dot
{c}_{t})px\label{HS}\\
H_{R}(x,p,t)  &  =\frac{1}{2}[2(\dot{\lambda}_{t}\lambda_{t}^{-1})\dot{\mu
}_{t}-\mu_{t}]x^{2}+(\dot{\lambda}_{t}\lambda_{t}^{-1})px \label{HR}%
\end{align}
where the dots $^{\cdot}$ signify differentiation $\tfrac{d}{dt}$ with respect
to $t$. (We invite the reader who is wondering by what magic we have obtained
these two formulas to have a sneak preview of Sections \ref{secham} and
\ref{secaff}.) The main observation we now make is that the Hamiltonian
$H_{R}$ lacks any term in $p^{2}$; it does not produce any kinetic energy. It
is a simple (but dull) exercise to show that while the Hamilton equations of
motion corresponding to respectively (\ref{HS}) and (\ref{HR}) are different
they lead to the same deformation of the initial disk $D(\varepsilon)$. Here
is an elementary example: consider the family of symplectic matrices%
\begin{equation}
S_{t}=%
\begin{pmatrix}
1 & t\\
0 & 1
\end{pmatrix}
; \label{free}%
\end{equation}
it corresponds to a particle with mass one freely moving along the $x$-axis
and it is thus the flow of the elementary Hamiltonian $H_{S}=\frac{1}{2}p^{2}%
$. The reduced Hamiltonian $H_{R}$, generating the flow $(R_{t})$, is given by%
\begin{equation}
H_{R}=-\frac{1}{2(1+t^{2})}x^{2}+\frac{tpx}{1+t^{2}}. \label{hr}%
\end{equation}
Let us now go one step further: while deforming the disk as just described, we
simultaneously move its center along an arbitrary curve $z_{t}=(x_{t},p_{t})$
starting from the origin at time $t=0$.\ Assuming that $z_{t}$ is continuously
differentiable we can view it as a Hamiltonian trajectory, and this in many
ways. The simplest choice is to take the \textquotedblleft translation
Hamiltonian\textquotedblright\ $H=p\dot{x}_{t}-x\dot{p}_{t}$ whose associated
Hamilton equations are $\dot{x}=\dot{x}_{t}$ and $\dot{p}=\dot{p}_{t}$. As
time elapses, the disk is being stretched and deformed while moving along
$z_{t}$, and at time $t$ it has become the ellipsoid $T(z_{t})R_{t}%
D(\varepsilon)$ where $T(z_{t})$ is the translation $z\longmapsto z+z_{t}$. An
absolutely not obvious fact is that the motion of this ellipsoid is
\textit{always} Hamiltonian! In fact, it corresponds to the inhomogeneous
quadratic Hamiltonian function%
\begin{equation}
H=H_{S}(x,p,t)+(a_{t}\dot{x}_{t}+b_{t}\dot{p}_{t})p-(c_{t}\dot{x}_{t}+d\dot
{p}_{t}) \label{hsa}%
\end{equation}
where the first term $H_{S}$ is given by formula (\ref{HS}) and represents the
deformation, while the second term corresponds to the motion. It is not
difficult to see that we can actually recover \textit{all} time-dependent
Hamiltonian functions of the type%
\[
H=\alpha_{t}x^{2}+\beta_{t}p^{2}+\gamma_{t}px+\delta_{t}p+\varepsilon_{t}x
\]
that is, all affine Hamiltonian flows by using \textquotedblleft chalkboard
motions\textquotedblright!

Let us now see what these manipulations become at the quantum level. Setting
$\varepsilon=\sqrt{\hbar}$ the disk $D(\sqrt{\hbar})=D(0,\sqrt{\hbar})$
corresponds to the Wigner ellipsoid of the standard Gaussian state%
\[
\phi_{0}(x)=(\pi\hbar)^{-1/4}e^{-x^{2}/2\hbar}.
\]
Now, to the family $(S_{t})$ one associates canonically a family of unitary
operators $(\widehat{S}_{t})$ (these are the metaplectic operators familiar
from quantum optics), and the deformation of $D(\sqrt{\hbar})$ by $S_{t}$
corresponds to the action of the metaplectic operator $\widehat{S}_{t}$ on
$\phi_{0}$. Now in general the new function $\widehat{S}_{t}\phi_{0}$ is
rather cumbersome to calculate. For instance, returning to the simple case of
free motion (\ref{free}) the function $\widehat{S}_{t}\phi_{0}$ is given by
the integral%
\begin{equation}
\widehat{S}_{t}\phi_{0}(x)=\left(  \tfrac{1}{2\pi i\hbar t}\right)  ^{1/2}%
\int_{-\infty}^{\infty}e^{\frac{i}{\hbar}\frac{(x-x^{\prime})^{2}}{2t}}%
\phi_{0}(x^{\prime})dx^{\prime} \label{stouffi}%
\end{equation}
(this formula actually holds for any $\phi_{0}$, not just the standard
coherent state); after some calculations involving the Fresnel formula for
Gaussian integrals one finds that%
\begin{equation}
\widehat{S}_{t}\phi_{0}(x)=\frac{1}{(\pi\hbar)^{1/4}}\frac{1}{\sqrt{1+it}}%
\exp\left(  -\frac{x^{2}}{2(1+it)\hbar}\right)  \label{stouffibis}%
\end{equation}
which is well-known in the literature on coherent states \cite{Littlejohn}. If
we now replace as in the geometric discussion above $(S_{t})$ with $(R_{t})$
we will have to replace $(\widehat{S}_{t})$ with the corresponding family
$(\widehat{R}_{t})$ of metaplectic operators. It turns out that these are
quite generally obtained from formula (\ref{reducedbis}) by%
\[
\widehat{R}_{t}\phi_{0}(x)=e^{\frac{i}{2\hbar}tx^{2}}(1+t^{2})^{-1/2}\phi
_{0}((1+t^{2})^{-1/2}x)
\]
which leads to
\begin{equation}
\widehat{R}_{t}\phi_{0}(x)=\frac{i^{\phi(t)}}{(\pi\hbar)^{1/4}\sqrt{1+t^{2}}%
}\exp\left(  -\frac{x^{2}}{2\hbar(1+t^{2})}\right)  . \label{stouffiter}%
\end{equation}
This is exactly formula (\ref{stouffibis}) up to $i^{\phi(t)}$ where $\phi(t)$
is a phase coming from the argument of $1+it$.. We have thus recovered the
propagation formula for the standard coherent state without any calculation of
integrals at all. In fact, this procedure works as well for arbitrary families
$(S_{t})$: if $S_{t}$ is given by (\ref{st}) with $a_{t}d_{t}-b_{t}c_{t}=1$
and then applying the corresponding metaplectic operator $\widehat{S}_{t}$ to
$\phi_{0}$ yields \cite{Littlejohn}%
\begin{equation}
\widehat{S}_{t}\phi_{0}(x)=\frac{1}{(\pi\hbar)^{1/4}}\frac{i^{\phi(t)}}%
{\sqrt{a_{t}+ib_{t}}}\exp\left(  -\frac{(d_{t}-ic_{t})x^{2}}{2(a_{t}%
+ib_{t})\hbar}\right)  . \label{stex}%
\end{equation}
We obtain the same result by applying the metaplectic version of
(\ref{reducedbis}), and this immediately yields%
\[
\widehat{R}_{t}\phi_{0}(x)=e^{-\frac{i}{2\hbar}\mu_{t}x^{2}}\lambda_{t}%
^{-1/2}\phi_{0}(\lambda_{t}x)
\]
where $\lambda_{t}=(a_{t}^{2}+b_{t}^{2})^{-1/2}$ and $\mu_{t}=a_{t}c_{t}%
+b_{t}d_{t}$ and this is seen to coincide with the expression above after some
trivial calculations, that is we have
\begin{equation}
\widehat{S}_{t}\phi_{0}=\widehat{R}_{t}\phi_{0}. \label{strt}%
\end{equation}
It should now be remarked that the time evolution of a Gaussian is -- as is
the evolution of any wavefunction -- governed by a Schr\"{o}dinger equation
(at least in a nonrelativistic setting); for instance for the free particle
considered above the function $\psi=\widehat{S}_{t}\phi_{0}$ satisfies%
\[
i\hbar\frac{\partial\psi}{\partial t}=-\frac{\hbar^{2}}{2}\frac{\partial
^{2}\psi}{\partial x^{2}}%
\]
as can be verified by a direct calculation using the explicit formula
(\ref{stouffibis}); the operator $\widehat{H}=-(\hbar^{2}/2)(\partial
^{2}/\partial x^{2})$ appearing in this equation is in fact the quantization
of the free-particle Hamiltonian function $H=p^{2}/2$. However we have seen
above (\textit{cf.} the equality (\ref{strt})) that we also have
$\psi=\widehat{R}_{t}\phi_{0}$; since the operators $\widehat{R}_{t}$
correspond to the flow $(R_{t})$ generated by the Hamiltonian function
(\ref{hr}), it seems plausible that the function $\psi$ should satisfy the
Schr\"{o}dinger equation corresponding to the quantization $\widehat{H}_{R}$
of this function, that is%
\[
i\hbar\frac{\partial\psi}{\partial t}=-t\left(  \frac{2+t^{2}}{1+t^{2}%
}\right)  x^{2}-\frac{i\hbar t}{1+t^{2}}\frac{1}{2}\left(  x\frac{\partial
\psi}{\partial x}+x\frac{\partial\psi}{\partial x}\right)  .
\]
It turns out that this guess is correct.

\subsection{Why ellipsoids are so useful in classical and quantum mechanics}

Mathematical points do not have any operative meaning in physics, be it
classical or quantum. Points live in the Platonic realm and are
epistemologically inaccessible. As Gazeau \cite{ga18} jokingly notes
\textquotedblleft...\textit{nothing is mathematically exact from the physical
point of view}\textquotedblright. What is however accessible to us are gross
approximations, like the chalk dots on a blackboard to take one naive example.
On a slightly more elaborate level, ellipsoids are good candidates as
substitutes for points. Consider for instance a particle moving in the
configuration space $\mathbb{R}_{x}^{n}$. We perform a succession of
simultaneous position and momentum measurements and find a cloud of points
concentrated in a small region of the phase space $\mathbb{R}_{x}^{n}%
\times\mathbb{R}_{p}^{n}$. Position and momentum measurements lead to a cloud
of points, and we can use a method familiar from multivariate statistical
analysis to associate to our cloud a phase space ellipsoid. It works as
follows: after having eliminated possible outliers, we associate a phase space
ellipsoid $\Omega$ of minimum volume containing the convex hull of the
remaining set of points. This is the \textquotedblleft John--L\"{o}wner
ellipsoid\textquotedblright\ \cite{Ball,jo48,kumar,sch08} which plays an
extremely important role not only in statistics \cite{silver,stat}, but also
in many other related and unrelated disciplines (\textit{e.g. }convex geometry
and optimization, optimal design, computational geometry, computer graphics,
and pattern recognition). Since (ideally) the precision of measurements can be
arbitrarily increased, the volume of $\Omega$ can become as small as we want:
for every $\varepsilon>0$ we can make a sequence of measurements such that
$\operatorname*{Vol}(\Omega)<\varepsilon$. What about quantum systems? We
again perform measurements leading to a plot of points in phase space. But
while we can arbitrarily decrease the volume of the John--L\"{o}wner ellipsoid
$\Omega$ by increasing the precision of measurements in the classical case,
this does not work out in quantum mechanics because of the uncertainty
principle (for the precision limits in quantum metrology see the recent review
article \cite{haase}). Suppose in fact that we project $\Omega$ on the planes
of conjugate variables $\mathcal{P}_{j}=\mathcal{P}(x_{j},p_{j})$. We thus
obtain $n$ ellipses $\omega_{1},\omega_{2},...,\omega_{n}$ contained in the
planes $\mathcal{P}_{1},\mathcal{P}_{2},...,\mathcal{P}_{n}$ and the
Heisenberg inequalities imply that we must have
\begin{equation}
\operatorname*{Area}(\omega_{j})\gtrsim\hbar. \label{vol2}%
\end{equation}
This does not however lead to any estimate on the volume of the ellipsoid
$\Omega$ as one is tempted to believe by inference from the classical case
\cite{ro18}; in particular it does \textit{not} imply that phase\ space is
coarse grained by minimum cells with volume $\sim\hbar^{n}$. Suppose indeed
the volume of the smallest John--L\"{o}wner ellipsoid $\Omega$ is, say,
\begin{equation}
\operatorname*{Vol}(\Omega)=\frac{\pi^{n}}{n!}\hbar^{n} \label{vol3}%
\end{equation}
which is the volume of a phase space ball with radius $\sqrt{\hbar}$. The
projection of this ball on each of the planes $\mathcal{P}_{j}$ of conjugate
variables is a disk with area $\pi\hbar=\frac{1}{2}h$ in conformity with
(\ref{vol2}). However, if we choose for $\Omega$ any ellipsoid
\begin{equation}
\Omega:\sum_{j=1}^{n}\frac{x_{j}^{2}+p_{j}^{2}}{R_{j}^{2}}\leq1 \label{ell1}%
\end{equation}
the equality (\ref{vol3}) will still hold provided that $R_{1}^{2}R_{2}%
^{2}\cdot\cdot\cdot R_{n}^{2}=\hbar^{2}$ but for this it is not necessary at
all that the projections of $\Omega$ on all the planes $\mathcal{P}_{j}$ have
area at least $\pi R_{j}^{2}=\frac{1}{2}h$. For instance, in the case $n=2$
the ellipsoid
\begin{equation}
\frac{x_{1}^{2}+p_{1}^{2}}{N\hbar}+\frac{x_{2}^{2}+p_{2}^{2}}{N^{-1}\hbar}%
\leq1 \label{ell2}%
\end{equation}
indeed has volume $\pi^{2}\hbar^{2}/2$ for every value of $N>0$ and its
projection on the $\mathcal{P}_{1}$ plane has area $N\pi\hbar$ but its
projection on the $\mathcal{P}_{2}$ plane has area $N^{-1}\pi\hbar$ so the
uncertainty principle is violated for large $N$. This discussion shows that
the volume condition $\operatorname*{Vol}(\Omega)\sim\hbar^{n}$ is not
sufficient to describe a phase space coarse graining; it does not allow us to
tell whether a phase space ellipsoid (or more general phase space domains) is
in compliance with the most basic feature of quantum mechanics, the
uncertainty principle (which it would be better to call the
\textit{indeterminacy principle}).

\subsection{Notation and terminology}

We will equip $\mathbb{R}^{2n}=T^{\ast}\mathbb{R}$ with the standard
symplectic structure $\sigma=\sum_{j=1}^{n}dp_{j}\wedge dx_{j}$ that is
\[
\sigma(z,z^{\prime})=p\cdot x^{\prime}-p^{\prime}\cdot x
\]
where $z=(x,p)$, $z^{\prime}=(x^{\prime},p^{\prime})$. In matrix notation
$\sigma(z,z^{\prime})=(z^{\prime})^{T}Jz$ where%
\[
J=%
\begin{pmatrix}
0_{n\times n} & I_{n\times n}\\
-I_{n\times n} & 0_{n\times n}%
\end{pmatrix}
.
\]
The symplectic group of $\mathbb{R}^{2n}$ is denoted by $\operatorname*{Sp}%
(n)$; it consists of all linear automorphisms of $\mathbb{R}^{2n}$ such that
$S^{\ast}\sigma=\sigma$, that is $\sigma(Sz,Sz^{\prime})=\sigma(z,z^{\prime})$
for all $z,z^{\prime}\in\mathbb{R}^{2n}$. Working in the canonical basis
$\operatorname*{Sp}(n)$ is identified with the group of all real $2n\times2n$
matrices $S$ such that $S^{T}JS=J$ (or, equivalently, $SJS^{T}=J$).

\section{Hamiltonian Flows\label{sechamn}}

We review some results from symplectic mechanics, that is Hamiltonian
mechanics expressed in the language of symplectic geometry; it is mainly
concerned with geometric properties and ultimately cooks down to the study of
symplectic isotopies. See for instance \cite{feng,Wang} for the implementation
of symplectic algorithms allowing the numerical resolution of Hamiltonian
systems. 

\subsection{Symplectic matrices}

It is convenient to fix a symplectic basis of $\mathbb{R}^{2n}$ and to
identify the symplectic automorphisms of $(\mathbb{R}^{2n},\sigma)$ with their
matrices in that basis. We will mainly work in the canonical basis (which is
both orthogonal and symplectic) and write symplectic matrices in block-form%
\begin{equation}
S=%
\begin{pmatrix}
A & B\\
C & D
\end{pmatrix}
\label{1sab}%
\end{equation}
where the submatrices $A,B,C,D$ all have same dimension $n\times n$. The
condition $S^{T}JS=J$ can then be expressed as conditions on these
submatrices; for instance
\begin{equation}
A^{T}D-C^{T}B=I_{\mathrm{d}}\text{ , }A^{T}C=CA^{T}\text{ , }B^{T}D=D^{T}B
\label{1cond}%
\end{equation}
or
\begin{equation}
AD^{T}-BC^{T}=I_{\mathrm{d}}\text{ , }AB^{T}=BA^{T}\text{ , }CD^{T}=DC^{T}.
\label{2cond}%
\end{equation}
Also, the inverse of $S$ is explicitly given by
\begin{equation}
S^{-1}=%
\begin{pmatrix}
D^{T} & -B^{T}\\
-C^{T} & A^{T}%
\end{pmatrix}
. \label{1inv}%
\end{equation}
See \cite{Birk}, \S 2.1, for details and additional material.

The symplectic group is generated by the set of all matrices $J$, $V_{-P}$ and
$M_{L}$ where
\begin{equation}
V_{-P}=%
\begin{pmatrix}
I & 0\\
P & I
\end{pmatrix}
\text{ , }M_{L}=%
\begin{pmatrix}
L^{-1} & 0\\
0 & L^{T}%
\end{pmatrix}
\end{equation}
with $P=P^{T}$ and $\det L\neq0$.

\subsection{The group $\operatorname*{Ham}(n)$}

Let $I_{T}$ be the closed interval $[-T,T]$ with $T>0$ (or $T=\infty$). A
function $H\in C^{\infty}(\mathbb{R}^{2n}\times I_{T},\mathbb{R})$ will be
called a \textquotedblleft Hamiltonian\textquotedblright; whenever necessary
it will be convenient to assume that there exists a compact subset $K$ of
$\mathbb{R}^{2n}$ such that the support of $H$ satisfies $\operatorname*{supp}%
H(\cdot,t)$ is contained in $K$ for all $t\in I_{T}$. This apparently
restrictive assumption avoids problems arising with Hamiltonians defined on
open manifolds which can have bad behavior at infinity, as discussed in
\cite{Polter}, \S 1.3 (for instance, solutions to Hamilton's equations can
blow up at finite time). When applied, this assumption guarantees that the
flow generated by the Hamiltonian vector field $X_{H}=J\partial_{z}H$ exists
for all $t\in I_{T}$. Some afterthought shows that this condition is after all
not too stringent. Assume for instance that we want to deal with the standard
physical Hamiltonian%
\[
H(x,p,t)=\frac{1}{2}|p|^{2}+V(x,t).
\]
The latter is never compactly supported. But for all (or most) practical
purposes we want to study the Hamilton equations $\dot{x}=p$, $\dot
{p}=-\partial_{x}V(x,t)$ in a bounded domain $D$ of phase space. Choosing a
compactly supported cutoff function $\chi\in C_{0}^{\infty}(\mathbb{R}^{2n})$
such that $\chi(z)=1$ for $z\in\overline{D}$ the solutions $t\longmapsto(x,p)$
of the Hamilton equations for $\chi H$ with initial value $z_{0}\in D$ are
just those of the initial problem $\dot{x}=p$, $\dot{p}=-\partial_{x}V(x,t)$
as long as $(x,p)$ remains in $D$.

Let $D$ be an open subset of $\mathbb{R}^{2n}$. A diffeomorphism
$f:D\longrightarrow\mathbb{R}^{2n}$ is called a symplectomorphism (or
\textquotedblleft canonical transformation\textquotedblright\ \cite{Arnold})
if $f^{\ast}\sigma=f$, that is, the Jacobian matrix $Df(z)$ is symplectic at
every $z\in D$:
\begin{equation}
(Df(z)^{T})JDf(z)=Df(z)J(Df(z)^{T})=J. \label{symplecto}%
\end{equation}
The symplectomorphisms of $D=\mathbb{R}^{2n}$ form a subgroup
$\operatorname*{Symp}(n)$ of the group $\operatorname*{Diff}(n)$ of all
diffeomorphisms of $\mathbb{R}^{2n}$: this easily follows from the relations
(\ref{symplecto}) above using the chain rule. The symplectic group
$\operatorname*{Sp}(n)$ is a subgroup of $\operatorname*{Symp}(n)$.

Let $H$ be a Hamiltonian in the sense above. The Hamiltonian vector field
$X_{H}=J\partial_{z}H$ is compactly supported and hence complete, so that
Hamilton's equations
\begin{equation}
\dot{z}(t)=J\partial_{z}H(z(t),t)\text{ \ , \ }z(0)=z_{0} \label{hameq1}%
\end{equation}
have a unique solution for every choice of initial point $z_{0}\in
\mathbb{R}^{2n}$. The time-dependent flow $(f_{t}^{H})$ generated by $X_{H}$
is the family of diffeomorphisms $z_{0}\longmapsto z(t)=f_{t}^{H}(z_{0})$
associating to $z_{0}$ the solution $z(t)$ at time $t\in I_{T}$. Each
$f_{t}^{H}$ is a symplectomorphism of $M$: $f_{t,t^{\prime}}^{H}%
\in\operatorname*{Symp}(n)$:%
\begin{equation}
(Df_{t}^{H}(z))^{T}JDf_{t}^{H}(z)=Df_{t}^{H}(z)J(Df_{t}^{H}(z)^{T})=J
\label{jac1}%
\end{equation}
where $Df_{t}^{H}(z)$ is the Jacobian matrix of $f_{t}^{H}$ calculated at $z$.

A remarkable fact is that composition and inversion of Hamiltonian flows also
yield Hamiltonian flows \cite{Birk,Birkbis,RMP,Polter,HZ}. Let $(f_{t}^{H})$
and $(f_{t}^{K})$ be determined by two Hamiltonian functions $H=H(z,t)$ and
$K=K(z,t)$; we have the following composition and inversion rules:

\begin{proposition}
\label{propham}(i) Let $H$ and $K$ be Hamiltonians; we have:%
\begin{align}
f_{t}^{H}f_{t}^{K}  &  =f_{t}^{H\#K}\text{ \ \ \textit{with} \ \ }%
H\#K(z,t)=H(z,t)+K((f_{t}^{H})^{-1}(z),t)\label{ch1}\\
(f_{t}^{H})^{-1}  &  =f_{t}^{\bar{H}}\text{ \ \ \textit{with} \ \ }\bar
{H}(z,t)=-H(f_{t}^{H}(z),t)\label{ch2}\\
f_{t}^{H+K}  &  =f_{t}^{H}f_{t}^{K^{\prime}}\text{ \ \textit{with} }K^{\prime
}(z,t)=K(f_{t}^{H}(z),t). \label{ch3}%
\end{align}
(ii) The composition law $\#$ defined by (\ref{ch1}) is associative: if $L$ is
a third Hamiltonian then
\begin{equation}
(H\#K)\#L=H\#(K\#L). \label{hkl}%
\end{equation}
(iii) For every $g\in\operatorname*{Symp}(n)$ we have the conjugation
property
\begin{equation}
g^{-1}f_{t}^{H}g=f_{t}^{H\circ g}. \label{conj}%
\end{equation}

\end{proposition}

\begin{proof}
The proofs of these formulas are based on the transformation property
$X_{g_{\ast}H}=g_{\ast}X_{H}$ of Hamiltonian vector fields for $g\in
\operatorname*{Symp}(n)$ \cite{Birk,RMP,HZ,Polter}. The associativity
(\ref{hkl}) follows from the associativity of the composition of mappings in a
same space: $(f_{t}^{H}f_{t}^{K})f_{t}^{L}=f_{t}^{H}(f_{t}^{K}f_{t}^{L})$
hence $f_{t}^{(H\#K)\#L}=f_{t}^{H\#(K\#L)}$.
\end{proof}

Notice that even when $H$ and $K$ are time-independent Hamiltonians the
functions $H\#K$ and $\bar{H}$ are generally time-dependent. Formula
(\ref{conj}) is often expressed in physics by saying that \textquotedblleft
Hamilton's equations are covariant under canonical
transformations\textquotedblright\ \cite{Arnold}.

Let $f\in\operatorname*{Diff}(n)$ such that $\ f=f_{t_{0}}^{H}$ for some
Hamiltonian function $H$ and time $t_{0}$. We will say that $f$ is a
\emph{Hamiltonian symplectomorphism}. (Rescaling time if necessary one can
always assume $t_{0}=1$.)

As immediately follows from formulas (\ref{ch1}), (\ref{ch2}), and
(\ref{conj}) Hamiltonian symplectomorphisms form a normal subgroup
$\operatorname*{Ham}(n)$ of $\operatorname*{Symp}(n)$. A somewhat surprising
fact is that every continuous path of Hamiltonian symplectomorphisms passing
through the identity is the phase flow of a Hamiltonian function; this was
first proved by Banyaga \cite{Banyaga} in the very general context of
symplectic manifolds; see \cite{Wang,RMP} for elementary proofs:

\begin{proposition}
\label{thm1} Let $(f_{t})$ be a smooth one-parameter family of Hamiltonian
symplectomorphisms such that $f_{0}=I_{\mathrm{d}}$. Then $(f_{t})$ is the
flow determined by the Hamiltonian function
\begin{equation}
H(z,t)=-\int_{0}^{1}\sigma(\dot{f}_{t}f_{t}^{-1}(\lambda z),z)d\lambda
\label{hzt}%
\end{equation}
where $\dot{f}_{t}=df_{t}/dt$.
\end{proposition}

We will call a smooth path $(f_{t})$, $t\in I_{T}$, in $\operatorname*{Ham}%
(n)$ joining the identity to some element $f\in\operatorname*{Ham}(n)$ a
\emph{Hamiltonian isotopy}.

\subsection{Time-dependent quadratic Hamiltonians\label{secham}}

Quadratic Hamiltonian functions are widely used both in theoretical and
practical situations. They are easy to manipulate in the time-independent case
and lead to rich structures (see the excellent review by Combescure and Robert
\cite{coro06}). Here we study general quadratic Hamiltonians with
time-dependent coefficients.

Let $S\in\operatorname*{Sp}(n)$. Since $\operatorname*{Sp}(n)$ is arcwise
connected we can find a $C^{1}$ path $t\longmapsto S_{t}$, $t\in I_{T}$ in
$\operatorname*{Sp}(n)$ joining the identity to $S$. We will write $S_{t}$ in
$n\times n$ block-matrix form%
\begin{equation}
S_{t}=%
\begin{pmatrix}
A_{t} & B_{t}\\
C_{t} & D_{t}%
\end{pmatrix}
\text{ \ , \ }S_{0}=I_{\mathrm{d}}. \label{st1}%
\end{equation}

\begin{proposition}
\label{propabcd}(i) $(S_{t})$ is the phase flow determined by the quadratic
Hamiltonian%
\begin{equation}
H(z,t)=-\frac{1}{2}J\dot{S}_{t}S_{t}^{-1}z^{2}; \label{hamzo}%
\end{equation}
(ii) the latter is explicitly given by
\begin{equation}
H=\tfrac{1}{2}(\dot{D}_{t}C_{t}^{T}-\dot{C}_{t}D_{t}^{T})x^{2}-(\dot{D}%
_{t}A_{t}^{T}-\dot{C}_{t}B_{t}^{T})p\cdot x+\tfrac{1}{2}(\dot{B}_{t}A_{t}%
^{T}-\dot{A}_{t}B_{t}^{T})p^{2}. \label{hamabcd}%
\end{equation}

\end{proposition}

\begin{proof}
\textit{(i)} In view of formula (\ref{hzt}) we have%
\begin{equation}
H(z,t)=-\int_{0}^{1}\sigma\left(  \dot{S}_{t}S_{t}^{-1}(\lambda z),z\right)
d\lambda; \label{hamzobis}%
\end{equation}
formula (\ref{hamzo}) follows since%
\[
\sigma\left(  \dot{S}_{t}S_{t}^{-1}(\lambda z),z\right)  =\lambda J\dot{S}%
_{t}S_{t}^{-1}z\cdot z.
\]
\textit{(ii)} The inverse of $S_{t}$ is given by%
\begin{equation}
S_{t}^{-1}=%
\begin{pmatrix}
D_{t}^{T} & -B_{t}^{T}\\
-C_{t}^{T} & A_{t}^{T}%
\end{pmatrix}
\end{equation}
and hence%
\begin{equation}
J\dot{S}_{t}S_{t}^{-1}=%
\begin{pmatrix}
\dot{C}_{t}D_{t}^{T}-\dot{D}_{t}C_{t}^{T} & \dot{D}_{t}A_{t}^{T}-\dot{C}%
_{t}B_{t}^{T}\\
\dot{B}_{t}C_{t}^{T}-\dot{A}_{t}D_{t}^{T} & \dot{A}_{t}B_{t}^{T}-\dot{B}%
_{t}A_{t}^{T}%
\end{pmatrix}
; \label{jstst}%
\end{equation}
formula (\ref{hamabcd}) now follows from (\ref{hamzo}).
\end{proof}

In particular, if $S_{t}=e^{tX}$ with $X\in\mathfrak{sp}(n)$ (the symplectic
Lie algebra) one recovers the usual formula $H=-\frac{1}{2}JXz^{2}$.

The group $U(n)=\operatorname*{Sp}(2n)\cap O(2n,\mathbb{R})$ of symplectic
rotations can be identified with the unitary group $U(n,\mathbb{C})$ via the
embedding
\[
U(n,\mathbb{C})\ni X+iY\longmapsto%
\begin{pmatrix}
X & Y\\
-Y & X
\end{pmatrix}
\in U(n);
\]
notice that $X$ and $Y$ must satisfy the conditions
\begin{align}
XX^{T}+YY^{T}  &  =I_{\mathrm{d}}\text{ \ , \ }XY^{T}-YX^{T}=0\label{xxyy1}\\
X^{T}X+Y^{T}Y  &  =I_{\mathrm{d}}\text{ \ , \ }X^{T}Y-Y^{T}X=0. \label{xxyy2}%
\end{align}
Suppose that the symplectic isotopy consists of symplectic rotations
\begin{equation}
U_{t}=%
\begin{pmatrix}
X_{t} & Y_{t}\\
-Y_{t} & X_{t}%
\end{pmatrix}
\text{ \ , \ }U_{0}=I_{\mathrm{d}}. \label{utxy}%
\end{equation}
Proposition \ref{propabcd} implies:

\begin{corollary}
\label{coro3}(i) The family $(U_{t})$ of symplectic rotations (\ref{utxy}) is
the flow generated by the quadratic Hamiltonian%
\begin{equation}
H_{U}=\tfrac{1}{2}Z_{t}x^{2}+\tfrac{1}{2}Z_{t}p^{2} \label{hamabcdxy}%
\end{equation}
where
\begin{equation}
Z_{t}=\dot{Y}_{t}X_{t}^{T}-\dot{X}_{t}Y_{t}^{T}=Z^{T}. \label{Z}%
\end{equation}
(ii) Conversely, every every Hamiltonian (\ref{hamabcdxy}) such that $Z=Z^{T}$
generates a flow $(U_{t})$ with $U_{t}\in U(n)$ and thus uniquely determines
$X_{t},Y_{t}$ such that (\ref{Z}) holds.
\end{corollary}

\begin{proof}
\textit{(i)} In view of formula (\ref{hamzo}) $(U_{t})$ is generated by the
Hamiltonian $H_{U}(z,t)=-\frac{1}{2}J\dot{U}_{t}U_{t}^{-1}z^{2}$. Noting that
$J$ commutes with both $U_{t}$ and $\dot{U}_{t}$ it follows that
$H_{U}(Jz,t)=H_{U}(z,t)$ and formula (\ref{hamabcdxy}) follows from
(\ref{hamabcd}) since the cross terms are $(\dot{X}X^{T}+\dot{Y}Y^{T})px=0$.
The matrix $Z_{t}$ is symmetric: we have $X_{t}Y_{t}^{T}-Y_{t}X_{t}^{T}=0$
hence, differentiating with respect to $t$, $\dot{Y}_{t}X_{t}^{T}-\dot{X}%
_{t}Y_{t}^{T}=X_{t}\dot{Y}_{t}^{T}-Y_{t}\dot{X}_{t}^{T}$. \textit{(ii)} Let
$(U_{t})$ be the flow determine by $H_{U}$. Since $U_{t}\in\operatorname*{Sp}%
(n)$ for all $t$ it is sufficient to show that in addition $U_{t}J=JU_{t}$.
Let $D_{t}=%
\begin{pmatrix}
Z_{t} & 0\\
0 & Z_{t}%
\end{pmatrix}
$. Since $\dot{U}_{t}=JD_{t}U_{t}$ we have%
\[
\frac{d}{dt}(U_{t}J)=\left(  \frac{d}{dt}U_{t}\right)  J=JD_{t}(U_{t}J);
\]
similarly, since $D_{t}J=JD_{t}$
\[
\frac{d}{dt}(JU_{t})=(-JD_{t}J)JU_{t}=JD_{t}(JU_{t})
\]
hence $U_{t}J$ and $JU_{t}$ satisfy the same first order differential equation
with same initial value $U_{0}J=JU_{0}=J$ so we must have $U_{t}J=JU_{t}$.
\end{proof}

The method described above extends to the case of affine symplectic isotopies
without difficulty:

\begin{proposition}
\label{thm2}Let $(S_{t})$ be a symplectic isotopy in $\operatorname*{Sp}(n)$
and $I_{T}\ni t\longmapsto z_{t}$ a $C^{1}$ path in $\mathbb{R}^{2n}$ with
$z_{0}=0$. (i) The affine symplectic isotopy $(f_{t})$ defined by $f_{t}%
=S_{t}T(z_{t})$ is the phase flow determined by the Hamiltonian%
\begin{equation}
H(z,t)=-\frac{1}{2}J\dot{S}_{t}S_{t}^{-1}z^{2}+\sigma\left(  z,S_{t}\dot
{z}_{t}\right)  \label{hamzoter}%
\end{equation}
and that defined by $g_{t}=T(z_{t})S_{t}$ is
\begin{equation}
H(z,t)=-\frac{1}{2}J\dot{S}_{t}S_{t}^{-1}(z-z_{t})+\sigma\left(  z,\dot{z}%
_{t}\right)  . \label{hamzoterbis}%
\end{equation}
(ii) Conversely, every Hamiltonian function%
\begin{equation}
H(z,t)=\frac{1}{2}M(t)z^{2}+m(t)z \label{hzm}%
\end{equation}
with $M(t)=M(t)^{T}$ and $m(t)\in\mathbb{R}^{2n}$ depending continuously on
$t\in I_{T}$ can be rewritten in the form (\ref{hamzoter}) (or
(\ref{hamzoterbis})).
\end{proposition}

\begin{proof}
\textit{(i) }Formula (\ref{hamzoter}) follows from formula (\ref{hamzobis})
using the product formula (\ref{ch1}) for Hamiltonian flows with $H=-\frac
{1}{2}J\dot{S}_{t}S_{t}^{-1}z^{2}$ and $K=\sigma\left(  z,\dot{z}_{t}\right)
$, and noticing that $\sigma\left(  S_{t}^{-1}z,\dot{z}_{t}\right)
=\sigma\left(  z,S_{t}\dot{z}_{t}\right)  $. Formula (\ref{hamzoterbis}) is
proven likewise swapping $H$ and $K$. \textit{(ii) }Let $(S_{t})$ be the flow
determined by the homogeneous part $H_{0}(z,t)=\frac{1}{2}M(t)z^{2}$ of
$H(z,t)$. We have $\dot{S}_{t}=JM(t)S_{t}$ hence $H_{0}(z,t)=-\frac{1}{2}%
J\dot{S}_{t}S_{t}^{-1}z^{2}$. Set now%
\begin{equation}
z_{t}=\int_{0}^{t}S_{t^{\prime}}^{-1}Jm(t^{\prime})dt^{\prime}, \label{fthm}%
\end{equation}
that is $\dot{z}_{t}=S_{t}^{-1}Jm(t)$; the Hamilton equations for (\ref{hzm})
are%
\[
\dot{z}(t)=JM(t)z(t)+Jm(t)=\dot{S}_{t}S_{t}^{-1}z(t)+S_{t}\dot{z}_{t}%
\]
and are solved by $z(t)=S_{t}T(z_{t})z(0)$; it follows that the flow $(f_{t})$
determined by $H$ is given by $f_{t}=S_{t}T(z_{t})$ and we can thus rewrite
$H$ as (\ref{hamzoter}).
\end{proof}

Assume for instance that the coefficients $M$ and $m$ are time-independent and
$\det M\neq0$; the solution of Hamilton's equations
\[
\dot{z}(t)=JMz(t)+Jm\text{ \ , \ }z(0)=z_{0}%
\]
are given by
\begin{equation}
z(t)=e^{tJM}z_{0}+(JM)^{-1}(e^{tJM}-I)Jm. \label{burdet}%
\end{equation}
(If $M$ fails to be invertible, this formula remains formally correct,
expanding $e^{tJM}$ in a Taylor series \cite{burdet}.)

\section{Symplectic Capacities}

We denote by $B^{2n}(z_{0},R)$ the ball $|z-z_{0}|\leq R$ in $\mathbb{R}^{2n}%
$; we write $B^{2n}(0,R)=B^{2n}(R)$. Denoting by $T(z_{0})$ the translation
$z\longmapsto z+z_{0}$ we have $B^{2n}(z_{0},R)=T(z_{0})B^{2n}(R)$.

\subsection{Definition of a symplectic capacity}

It is well-known that Hamiltonian flows are volume preserving
\cite{Arnold,Polter}; this property is an easy consequence of the fact that
any Hamiltonian flow consists of symplectomorphisms and hence preserves the
successive powers $\sigma$, $\sigma\wedge\sigma$,...,$\sigma^{\wedge n}$ of
the symplectic form. This property is however not characteristic of
Hamiltonian flows, because any flow generated by a divergence-free vector
fields has this property. It however turns out that there exist quantities
whose preservation is characteristic of Hamiltonian flows (and, more
generally, of symplectomorphisms). These are the symplectic capacities of
subsets of phase space.

A (normalized, or intrinsic) symplectic capacity on $(\mathbb{R}^{2n},\sigma)$
assigns to every $\Omega\subset\mathbb{R}^{2n}$ a number $c(\Omega)\geq0$, or
$+\infty$, and must satisfy the following axioms \cite{HZ}:

\begin{description}
\item[(SC1)] \textbf{Monotonicity}: \textit{If} $\Omega\subset\Omega^{\prime}$
\textit{then} $c(\Omega)\leq c(\Omega^{\prime})$;

\item[(SC2)] \textbf{Symplectic invariance}: \textit{If} $f\in
\operatorname*{Symp}(n)$\textit{ then} $c(f(\Omega))=c(\Omega)$;

\item[(SC3)] \textbf{Conformality}: \textit{If} $\lambda\in\mathbb{R}$
\textit{then} $c(\lambda\Omega)=\lambda^{2}c(\Omega)$;

\item[(SC4)] \textbf{Non-triviality}: \textit{We have }$c(B^{2n}(R))=\pi
R^{2}=c(Z_{j}^{2n}(R))$ \textit{where} $Z_{j}^{2n}(R)$ \textit{is the
cylinder} $x_{j}^{2}+p_{j}^{2}\leq R^{2}$ in $\mathbb{R}^{2n}$.
\end{description}

The archetypical example of a symplectic capacity is the \textquotedblleft
Gromov width\textquotedblright\ \cite{HZ}\ defined by
\begin{equation}
c_{\min}(\Omega)=\sup_{f\in\operatorname*{Symp}(n)}\{\pi R^{2}:f(B^{2n}%
(R))\subset\Omega\}; \label{cmin}%
\end{equation}
that $c_{\min}$ indeed satisfies axiom (SC4) is a non-trivial property,
equivalent to Gromov's symplectic non-squeezing theorem \cite{Gromov}: if
there exists $f\in\operatorname*{Symp}(n)$ such that $f(B^{2n}(R))\subset
Z_{j}^{2n}(r)$ then $R\leq r$ (see \cite{Birk,FOOP,goluPR} for discussions of
Gromov's result). As the notation suggests, $c_{\min}$ is the smallest of all
symplectic capacities: $c_{\min}\leq c\leq c_{\max}$ where%
\begin{equation}
c_{\max}(\Omega)=\inf_{f\in\operatorname*{Symp}(n)}\{\pi R^{2}:f(\Omega
)\subset Z_{j}^{2n}(R)\} \label{cmax}%
\end{equation}
so that we have $c_{\min}\leq c\leq c_{\max}$ for every symplectic capacity
$c$ on $(\mathbb{R}^{2n},\sigma)$.

Let $c$ be a symplectic capacity on the phase plane $\mathbb{R}^{2}$. Then
$c(\Omega)=\operatorname*{Area}(\Omega)$ for every connected and simply
connected surface $\Omega$. In higher dimensions the symplectic capacity can
be finite while the volume is infinite: for instance the symplectic capacity
of a cylinder $Z_{j}^{2n}(R)$ is finite, whereas its volume is infinite. It
follows in fact from the monotonicity and non-triviality properties of a
symplectic capacity that
\[
B^{2n}(R)\subset\Omega\subset Z_{j}^{2n}(R)\mathit{\ }\Longrightarrow
c(\Omega)=\pi R^{2}%
\]
so $\Omega$ can have arbitrarily large volume (even infinite). Symplectic
capacities are not related to volume when $n>1$; for instance if $\Omega$ and
$\Omega^{\prime}$ are disjoint we do not in general have $c(\Omega\cup
\Omega^{\prime}\mathcal{)}=c(\Omega)+c(\Omega^{\prime})$.

There exist infinitely many symplectic capacities, but they all agree on phase
space ellipsoids. The symplectic capacity of an ellipsoid%
\[
\Omega=\{z:M(z-z_{0})^{2}\leq R^{2}\}
\]
($M=M^{T}>0$) is calculated as follows. Recall Williamson's symplectic
diagonalization theorem \cite{Birk,HZ}: for every positive-definite symmetric
real $2n\times2n$ matrix $M$ there exists $S\in\operatorname*{Sp}(n)$ such
that
\begin{equation}
S^{T}MS=%
\begin{pmatrix}
\Lambda & 0\\
0 & \Lambda
\end{pmatrix}
\text{ } \label{will5}%
\end{equation}
where\ $\Lambda=\operatorname*{diag}(\lambda_{1}^{\sigma},...,\lambda
_{j}^{\sigma})$ is the diagonal matrix whose diagonal entries $\lambda
_{j}^{\sigma}$ are the symplectic eigenvalues of $M$: $\lambda_{j}^{\sigma}>0$
and the numbers $\pm i\lambda_{j}$ are the eigenvalues of $JM$ (that these
eigenvalues are indeed of the type $\pm i\lambda_{j}$ follows from the fact
that $JM$ has the same eigenvalues as the antisymmetric matrix $M^{1/2}%
JM^{1/2}$). This allows us to put the equation of the ellipsoid $\Omega$ in
the diagonal form%
\[
\sum_{j}\lambda_{j}^{\sigma}((x_{j}-x_{0,j})^{2}+(p_{j}-p_{0,j})^{2})\leq
R^{2}%
\]
and it is then easy to see \cite{Birk,Birkbis,goluPR}, using Gromov's
non-squeezing theorem that
\begin{equation}
c_{\min}(\Omega)=c_{\max}(\Omega)=\pi R^{2}/\lambda_{\max} \label{cw}%
\end{equation}
where $\lambda_{\max}$ is the largest symplectic eigenvalue of $M$. It follows
that $c(\Omega)=\pi R^{2}/\lambda_{\max}$ for every symplectic capacity $c$.

\subsection{A continuity property}

Let $\Omega$ and $\Omega^{\prime}$ be two nonempty compact subsets of
$\mathbb{R}^{2n}$. The numbers%
\[
d_{1}(\Omega,\Omega^{\prime})=\sup_{z\in\Omega}d(z,\Omega^{\prime})\text{ \ ,
\ }d_{2}(\Omega,\Omega^{\prime})=\sup_{z^{\prime}\in\Omega^{\prime}%
}d(z^{\prime},\Omega)\text{ }%
\]
are called, respectively, the directed Hausdorff distance from $\Omega$ to
$\Omega^{\prime}$ and from $\Omega^{\prime}$ to $\Omega$. The number
\begin{equation}
d_{\mathrm{H}}(\Omega,\Omega^{\prime})=\max(d_{1}(\Omega,\Omega^{\prime
}),d_{2}(\Omega,\Omega^{\prime})) \label{Haus1}%
\end{equation}
is called the \textit{Hausdorff distance} of $\Omega$ and $\Omega^{\prime}$.
(In some texts one defines this distance as the sum $d_{1}(\Omega
,\Omega^{\prime})+d_{2}(\Omega,\Omega^{\prime})$; both choices of course lead
to the same topology since the metrics are equivalent). The Hausdorff distance
is a metric on the set $\mathcal{K}(2n)$ of all nonempty compact subsets of
$\mathbb{R}^{2n}$ and $(\mathcal{K}(2n),d_{\mathrm{H}})$ is a complete metric
space \cite{federer}. We have $d(\Omega,\Omega^{\prime})<\infty$ since
$\Omega$ and $\Omega^{\prime}$ are bounded. If $0\in\Omega$, for every
$\varepsilon>0$ there exists $\delta>0$ such that \cite{abbobis}%
\begin{equation}
d_{\mathrm{H}}(\Omega,\Omega^{\prime})<\delta\Longrightarrow(1-\varepsilon
)\Omega\subset\Omega^{\prime}\subset(1+\varepsilon)\Omega. \label{Haus2}%
\end{equation}

Using (\ref{Haus2}) one shows that

\begin{proposition}
\label{PropHaus}Let $c$ be a symplectic capacity on $(\mathbb{R}^{2n},\sigma
)$. The restriction of $c$ to the set of all convex compact subsets equipped
with the Hausdorff distance is continuous. That is, for every $\varepsilon>0$
there exists $\delta>0$ such that if $\Omega$ and $\Omega^{\prime}$ are convex
and compact then
\begin{equation}
d_{H}(\Omega,\Omega^{\prime})<\delta\Longrightarrow|c(\Omega)-c(\Omega
^{\prime})|<\varepsilon\label{ccont}%
\end{equation}

\end{proposition}

\noindent(see \textit{e.g.} \cite{duff}, p.376).

The Hausdorff distance is used in pattern recognition and computer vision, and
also plays an essential role in medical imaging \cite{aziz}.

\subsection{Moving phase space ellipsoids}

We now give a fundamental characterization of symplectomorphisms initially due
to Ekeland and Hofer \cite{ekho1}.

Let $\Omega$ be a phase space ellipsoid as above; in \cite{ekho1,HZ} it is
proven that the only $C^{1}$ mappings that preserve the symplectic capacities
of all ellipsoids in $(\mathbb{R}^{2n},\sigma)$ are either symplectic or
antisymplectic (an antisymplectic mapping $F:\mathbb{R}^{2n}\longrightarrow
\mathbb{R}^{2n}$ is such that $F^{\ast}\sigma=-\sigma$; if $F$ is linear and
identified with its matrix this means that $F^{T}JF=-J$). In \cite{digopr14}
we have proven a refinement of this result in the linear case, by showing that
it is sufficient to consider a particular class of ellipsoids, called
symplectic balls. By definition, a symplectic ball is the image of a phase
space ball by an element of the inhomogeneous symplectic group
$\operatorname*{ISp}(n)$. It is thus an ellipsoid
\begin{equation}
B_{S}^{2n}(z_{0},R)=T(z_{0})SB^{2n}(R) \label{sball1}%
\end{equation}
or, equivalently
\begin{equation}
B_{S}^{2n}(z_{0},R)=ST(S^{-1}z_{0})B^{2n}(R). \label{sball2}%
\end{equation}
As follows from axioms (SC2) and (SC4) characterizing symplectic capacities, a
symplectic ball has symplectic capacity
\[
c(B_{S}^{2n}(z_{0},R))=c(B^{2n}(R))=\pi R^{2}.
\]

The proof of our refinement relies on the following algebraic result:

\begin{lemma}
\label{LemmaSymp}Let $F\in GL(2n,\mathbb{R})$. If $F^{T}M_{L}F\in
\operatorname*{Sp}(n)$ for every symplectic matrix
\begin{equation}
M_{L}=%
\begin{pmatrix}
L^{-1} & 0\\
0 & L
\end{pmatrix}
\text{ \ , \ }L=L^{T}>0 \label{mlf}%
\end{equation}
then $F$ is either symplectic or antisymplectic: $F^{T}JF=\pm J$.
\end{lemma}

The proof of this lemma is rather long and technical, we therefore refer to
\cite{digopr14}, \S 1.2, for a detailed argument. The particular symplectic
matrices (\ref{mlf}) will play an essential role in Section \ref{seciwa} where
we study the pre-Iwasawa factorization of general symplectic matrices. Notice
that the matrices (\ref{mlf}) do not form a subgroup of $\operatorname*{Sp}%
(n)$: we have $M_{L}M_{L^{\prime}}=M_{L^{\prime}L}$ but in general $L^{\prime
}L$ is not symmetric if $L$ and $L^{\prime}$ are.

\begin{proposition}
\label{propfund}(i) Assume that $K\in GL(2n,\mathbb{R})$ takes the symplectic
ball $B_{S}^{2n}(z_{0},R)$ to a symplectic ball $B_{S^{\prime}}^{2n}%
(z_{0}^{\prime},R)$ with the same radius. Then, $K$ is either symplectic or
antisymplectic. (ii) More generally, if $K$ takes every ellipsoid in
$\mathbb{R}^{2n}$ to an ellipsoid with the same symplectic capacity, then $K$
is symplectic or antisymplectic.
\end{proposition}

\begin{proof}
\textit{(i)} Since translations are symplectomorphisms we can assume $z_{0}=0$
so that the symplectic ball is just $SB^{2n}(R)$ and thus defined by the
inequality $|S^{-1}z|\leq R$. It follows that its image by $K$ is the set of
all $z\in\mathbb{R}^{2n}$ such that $|(KS)^{-1}z|\leq R$ that is
\[
((KS)^{-1})^{T}(KS)^{-1}z^{2}\leq R^{2}.
\]
If $KSB^{2n}(R)$ is a symplectic ball we must thus have
\[
((KS)^{-1})^{T}(KS)^{-1}=(K^{T})^{-1}(SS^{T})^{-1}K^{-1}\in\operatorname*{Sp}%
(n).
\]
Taking $F=K^{-1}$ then in view of Lemma \ref{LemmaSymp} the matrix $F$ and
hence $K$ must be either symplectic or antisymplectic.
\end{proof}

In the nonlinear case we have (\cite{ekho1}, Thm. 4):

\begin{proposition}
\label{propek1}Let $f:\mathbb{R}^{2n}\longrightarrow\mathbb{R}^{2n}$ be a
$C^{1}$ diffeomorphism such that $c(f(\Omega))=c(\Omega)$ for every ellipsoid
$\Omega\subset\mathbb{R}^{2n}$. Then either $f^{\ast}\sigma=\sigma$ or
$f^{\ast}\sigma=-\sigma$, that is $f$ is either a symplectomorphism or an anti-symplectomorphism.
\end{proposition}

In \cite{ekho1} (Thm. 5) Ekeland and Hofer prove the following nonlinear
version of Proposition \ref{propek1} using a mild differentiability requirement:

\begin{proposition}
\label{thmchalk}(i) Let $(f_{t})_{t\in I_{T}}$ be a family of $C^{1}$
diffeomorphisms $\mathbb{R}^{2n}\longrightarrow\mathbb{R}^{2n}$ such that
$f_{0}=I_{\mathrm{d}}$ and $c(f_{t}(\Omega))=c(\Omega)$ for every ellipsoid
$\Omega\subset\mathbb{R}^{2n}$ and $t\in I_{T}$. Then $(f_{t})$ is a
Hamiltonian flow. (ii) If the $f_{t}$ are affine mappings then $H$ is a
quadratic polynomial of the type (\ref{hzm}).
\end{proposition}

\begin{proof}
\textit{(i)} Let $z\in\mathbb{R}^{2n}$; we have
\[
(Df_{t}(z))^{T}JDf_{t}(z)=\pm J
\]
for all $t\in I_{T}$. Since $Df_{0}(z)=z$ and the mapping $t\longmapsto
Df_{t}(z)$ is continuous the only possible choice is%
\[
(Df_{t}(z))^{T}JDf_{t}(z)=J
\]
and hence the $f_{t}$ are symplectomorphisms; the conclusion now follows from
Proposition \ref{thm1}; the Hamiltonian \ is given by
\begin{equation}
H(z,t)=-\int_{0}^{1}\sigma(\dot{f}_{t}f_{t}^{-1}(\lambda z),z)d\lambda.
\end{equation}
\textit{(ii)} Immediately follows using Proposition \ref{thm2}.
\end{proof}

We urge the Reader to note that the assumption that symplectic capacities of
ellipsoids -- and not volumes -- are preserved is essential. As soon as $n>1$
the conclusions of Proposition \ref{thmchalk} for families of mappings which
are volume preserving are no longer true. Consider a divergence-free vector
field on $\mathbb{R}^{2n}$; by Liouville's theorem \cite{Arnold} the flow it
generates is certainly volume-preserving, but has no reason in general to
preserve symplectic capacities. The properties above are of a topological
nature, and show that general volume-preserving mappings cannot be
approximated in the $C^{0}$ topology by symplectomorphisms. In this context we
remark \cite{RMP} that Katok \cite{Katok} has shown that given two subsets
$\Omega$ and $\Omega^{\prime}$ with the same volume, then for every
$\varepsilon>0$ there exists $f\in\operatorname*{Symp}(n)$ such that
$\operatorname*{Vol}(f(\Omega)\setminus\Omega^{\prime})<\varepsilon$. Thus, an
arbitrarily large part of $\Omega$ can be symplectically embedded inside
$\Omega^{\prime}$ -- but not all of it! This again shows how different the
notions of volume conservation and symplectic capacity conservation are. 

\section{Symplectic Actions on Ellipsoids}

We introduce here the notion of local symplectic automorphisms; the
terminology will be justified in Section \ref{secmet} where we will show that
these automorphisms are the projections on $\operatorname*{Sp}(n)$ of local
metaplectic operators, \textit{i.e.} those which preserve the supports of
functions or tempered distributions.

\subsection{The local symplectic group\label{seclocal}}

Let $\ell$ be a Lagrangian subspace of the symplectic phase space
$(\mathbb{R}^{2n},\sigma)$: $\dim\ell=n$ and $\sigma$ vanishes identically on
$\ell$. Such a maximal isotropic subspace is also called a \textquotedblleft
Lagrangian plane\textquotedblright. One proves (\cite{Birk}, \S 2.2) that for
every $S\in\operatorname*{Sp}(n)$ there exist $S_{1},S_{2}\in
\operatorname*{Sp}(n)$ such that $S_{1}\ell\cap\ell=S_{2}\ell\cap\ell=0$ and
$S=S_{1}S_{2}$. Choosing for $\ell$ the momentum space $0\times\mathbb{R}^{n}$
the symplectic matrices $S_{1},S_{2}$ are of the type
\begin{equation}
S_{j}=%
\begin{pmatrix}
A_{j} & B_{j}\\
C_{j} & D_{j}%
\end{pmatrix}
\text{ \ , }\det B_{j}\neq0 \label{s12}%
\end{equation}
with $n\times n$ blocks. (Such symplectic matrices are called \textit{free.
}We will return to them in Section \ref{secmet}.) A straightforward
calculation leads to the factorization%
\begin{equation}
S_{j}=V_{-D_{j}B_{j}^{-1}}M_{B_{j}^{-1}}JV_{-B_{j}^{-1}A_{j}} \label{sdba}%
\end{equation}
where we define, for $P=P^{T}$ and $\det L\neq0$,
\begin{equation}
V_{-P}=%
\begin{pmatrix}
I & 0\\
P & I
\end{pmatrix}
\text{ , }M_{L}=%
\begin{pmatrix}
L^{-1} & 0\\
0 & L^{T}%
\end{pmatrix}
\label{vpml}%
\end{equation}
($DB^{-1}$ and $B^{-1}A$ are indeed symplectic due to the constraints imposed
on the blocks $A,B,C,D$ by the conditions $S^{T}JS=SJS^{T}=J$). It follows
that the set of all matrices $V_{-P}$ and $M_{L}$ together with the standard
symplectic matrix $J$ generate $\operatorname*{Sp}(n)$. Notice that these
matrices obey the product formulas%
\begin{equation}
V_{-P}V_{-P^{\prime}}=V_{-(P+P^{\prime})}\text{ \ , \ }M_{L}M_{L^{\prime}%
}=M_{L^{\prime}L}. \label{vpmlprod}%
\end{equation}
Let $\operatorname*{St}(\ell)$ be the stabilizer of $\ell$ in
$\operatorname*{Sp}(n)$: it is the subgroup of all $S\in\operatorname*{Sp}(n)$
such that $S\ell=\ell$. Of special importance for us is the case $\ell
=0\times\mathbb{R}^{n}$, we will write $\operatorname*{Sp}_{0}%
(n)=\operatorname*{St}(0\times\mathbb{R}^{n})$. It consists of all symplectic
block matrices with upper corner $B=0$. Since%
\[
V_{-P}M_{L}=%
\begin{pmatrix}
L^{-1} & 0\\
PL^{-1} & L^{T}%
\end{pmatrix}
\text{ \ , \ }M_{L}V_{-P}=%
\begin{pmatrix}
L^{-1} & 0\\
L^{T}P & L^{T}%
\end{pmatrix}
\]
the group $\operatorname*{Sp}_{0}(n)$ is generated by the symplectic matrices
$V_{-P}$ and $M_{L}$. It is thus the extension of the group $\{M_{L}:\det
L\neq0\}$ of symplectic rescalings by the group of symplectic shears
$\{V_{-P}:P=P^{T}\}$. Using the obvious identities
\begin{equation}
M_{L}V_{-P}=V_{-L^{T}PL}M_{L}\text{ \ },\text{ \ }V_{-P}M_{L}=M_{L}%
V_{-(L^{-1})^{T}PL^{-1}} \label{mlvp}%
\end{equation}
and
\begin{equation}
(V_{-P}M_{L})^{-1}=V_{-(L^{-1})^{T}PL^{-1}}M_{L^{-1}} \label{mlpinv}%
\end{equation}
we see that the group $\operatorname*{Sp}_{0}(n)$ in fact simply consists of
all products $V_{-P}M_{L}$ (or $M_{L}V_{-P}$). In particular, if
$S=V_{-P}M_{L}$ and $S^{\prime}=V_{-P^{\prime}}M_{L^{\prime}}$ we have%
\begin{equation}
S^{\prime}S^{-1}=V_{-P^{\prime}+(L^{-1}L^{\prime})^{T}P(L^{-1}L^{\prime}%
)}M_{L^{-1}L^{\prime}}. \label{xinvy}%
\end{equation}

The affine (or inhomogeneous) extension \cite{burdet}
\begin{equation}
\operatorname*{ISp}(n)=\operatorname*{Sp}(n)\ltimes\mathbb{R}^{2n}
\label{ispn}%
\end{equation}
of the symplectic group consists of all products $ST(z)=T(Sz)S$ where $T(z)$
is the translation operator in $\mathbb{R}^{2n}$ . Every element of
$\operatorname*{ISp}(n)$ can be written as either a product $ST(z_{0})$, or a
product $T(z_{0})S$. For each element of $\operatorname*{ISp}(n)$ this
factorization is unique: for example if $ST(z_{0})=S^{\prime}T(z_{0}^{\prime
})$ then $(S^{\prime})^{-1}S=T(z_{0}^{\prime}-z_{0})$ which is only possible
if $z_{0}^{\prime}=z_{0}$ and hence $S^{\prime}=S$. We call the subgroup
\begin{equation}
\operatorname*{ISp}\nolimits_{0}(n)=\operatorname*{Sp}\nolimits_{0}%
(n)\ltimes\mathbb{R}^{2n} \label{ispon}%
\end{equation}
the \textquotedblleft local inhomogeneous symplectic group\textquotedblright.
It consists of all affine symplectic transformations of the type%
\begin{equation}
T(z_{0})V_{-P}M_{L}=V_{-P}M_{L}T(M_{L^{-1}}V_{P}z_{0}). \label{pr1}%
\end{equation}
The main formulas are recapitulated in the table below:

\begin{center}%
\begin{tabular}
[c]{|l|l|}\hline
$V_{-P}V_{-P^{\prime}}=V_{-(P+P^{\prime})}$ & $M_{L}M_{L^{\prime}%
}=M_{L^{\prime}L}$\\\hline
$M_{L}V_{-P}=V_{-L^{T}PL}M_{L}$ & $V_{-P}M_{L}=M_{L}V_{-(L^{-1})^{T}PL^{-1}}%
$\\\hline
$(V_{-P}M_{L})^{-1}=V_{-(L^{-1})^{T}PL^{-1}}M_{L^{-1}}$ & $T(z_{0})V_{-P}%
M_{L}=V_{-P}M_{L}T(M_{L^{-1}}V_{P}z_{0})$\\\hline
\end{tabular}

\end{center}

\subsection{The pre-Iwasawa factorization\label{seciwa}}

We have seen in formula (\ref{sdba}) that every symplectic matrix
\begin{equation}
S=%
\begin{pmatrix}
A & B\\
C & D
\end{pmatrix}
\label{block}%
\end{equation}
with $\det B\neq0$ can be factorized as $S=V_{-P}M_{L}JV_{-P}$. The
pre-Iwasawa factorization generalizes this result; it says that every
$S\in\operatorname*{Sp}(n)$ can be written as a product of an element of a
subgroup of the local symplectic group $\operatorname*{Sp}_{0}(n)$ and of a
symplectic rotation. More precisely, writing $S\in\operatorname*{Sp}(n)$ in
block-matrix form ($n\times n$ blocks) there exist unique matrices $P=P^{T}$
and $L=L^{T}>0$ and $U_{X,Y}\in U(n)$ such that%
\begin{equation}
S=%
\begin{pmatrix}
I & 0\\
P & I
\end{pmatrix}%
\begin{pmatrix}
L^{-1} & 0\\
0 & L
\end{pmatrix}%
\begin{pmatrix}
X & Y\\
-Y & X
\end{pmatrix}
=V_{-P}M_{L}U_{X,Y}. \label{preiwa}%
\end{equation}
These matrices are given by%
\begin{align}
P  &  =(CA^{T}+DB^{T})(AA^{T}+BB^{T})^{-1}=P^{T}\label{pl1}\\
L  &  =(AA^{T}+BB^{T})^{-1/2}=L^{T}>0\label{pl2}\\
X  &  =(AA^{T}+BB^{T})^{-1/2}A\text{ \ },Y=(AA^{T}+BB^{T})^{-1/2}B.
\label{unixy}%
\end{align}

The proof of these formulas is purely computational; see
\cite{bera,Dutta,Birk,kbw}. It is clear that $L=L^{T}$ is positive definite;
that $P$ is also symmetric follows from the fact that the relation
$S=V_{-P}M_{L}U$ implies that $V_{-P}$ is symplectic which requires that
$P=P^{T}$. The uniqueness follows from the observation that if $V_{-P}%
M_{L}U_{X,Y}=V_{-P^{\prime}}M_{L^{\prime}}U_{X^{\prime},Y^{\prime}}$ then
\[
M_{L^{\prime}}^{-1}V_{P-P^{\prime}}M_{L}=U_{X^{\prime},Y^{\prime}}U_{X,Y}%
^{-1}=U_{X^{\prime\prime},Y^{\prime\prime}}%
\]
and thus%
\[%
\begin{pmatrix}
L^{\prime}L^{-1} & 0\\
(L^{\prime})^{-1}(P-P^{\prime})L & (L^{\prime})^{-1}L
\end{pmatrix}
=%
\begin{pmatrix}
X^{\prime\prime} & Y^{\prime\prime}\\
-Y^{\prime\prime} & X^{\prime\prime}%
\end{pmatrix}
\]
hence $P=P^{\prime}$ and $L^{\prime}=L$ since $L,L^{\prime}>0$.

Notice that, as a particular case, any dilation $M_{K}=%
\begin{pmatrix}
K^{-1} & 0\\
0 & K^{T}%
\end{pmatrix}
$, $\det K\neq0$, has the pre-Iwasawa factorization%
\begin{equation}%
\begin{pmatrix}
(K^{T}K)^{-1/2} & 0\\
0 & (K^{T}K)^{1/2}%
\end{pmatrix}%
\begin{pmatrix}
(K^{T}K)^{1/2}K^{-1} & 0\\
0 & (K^{T}K)^{1/2}K^{-1}%
\end{pmatrix}
. \label{KL}%
\end{equation}

Summarizing:

\begin{center}%
\begin{tabular}
[c]{|l|}\hline
The symplectic matrix (\ref{block}) has a unique factorization\\\hline
$S=RU$ where $R\in\operatorname*{Sp}_{0}(n)$ and $U\in U(n)$ are given
by:$\medskip$\\\hline
$R=%
\begin{pmatrix}
(AA^{T}+BB^{T})^{1/2} & 0\\
(CA^{T}+DB^{T})(AA^{T}+BB^{T})^{-1/2} & (AA^{T}+BB^{T})^{-1/2}%
\end{pmatrix}
\medskip$\\\hline
$U=%
\begin{pmatrix}
(AA^{T}+BB^{T})^{-1/2}A\text{ } & (AA^{T}+BB^{T})^{-1/2}B\\
-(AA^{T}+BB^{T})^{-1/2}B & (AA^{T}+BB^{T})^{-1/2}A\text{ }%
\end{pmatrix}
\medskip$\\\hline
\end{tabular}

\end{center}

When writing $S=RU$ we will call $R$ and $U$ respectively the \textit{local}
and the \textit{unitary} components of $S$. They are uniquely defined.

\subsection{Iwasawa factorization of a quadratic Hamiltonian}

Let $(S_{t})$ be a symplectic isotopy and
\begin{equation}
H(z,t)=-\frac{1}{2}J\dot{S}_{t}S_{t}^{-1}z^{2} \label{hst}%
\end{equation}
the associated Hamiltonian. Writing
\[
S_{t}=%
\begin{pmatrix}
A_{t} & B_{t}\\
C_{t} & D_{t}%
\end{pmatrix}
\]
the pre-Iwasawa factorization yields $S_{t}=R_{t}U_{t}$ where $R_{t}%
=V_{-P_{t}}M_{L_{t}}$ is given by
\begin{equation}
R_{t}=%
\begin{pmatrix}
L_{t}^{-1} & 0\\
P_{t}L_{t}^{-1} & L_{t}%
\end{pmatrix}
=%
\begin{pmatrix}
L_{t}^{-1} & 0\\
Q_{t} & L_{t}%
\end{pmatrix}
, \label{rt}%
\end{equation}
the symmetric $n\times n$ matrices $P_{t}$ and $L_{t}$ being calculated using
the formulas (\ref{pl1}) and (\ref{pl2}):%
\begin{align}
P_{t}  &  =(C_{t}A_{t}^{T}+D_{t}B_{t}^{T})(A_{t}A_{t}^{T}+B_{t}B_{t}^{T}%
)^{-1}\label{pl1bis}\\
L_{t}  &  =(A_{t}A_{t}^{T}+B_{t}B_{t}^{T})^{-1/2}\label{pl2bis}\\
Q_{t}  &  =(C_{t}A_{t}^{T}+D_{t}B_{t}^{T})(A_{t}A_{t}^{T}+B_{t}B_{t}%
^{T})^{-1/2}. \label{pl3bis}%
\end{align}
Similarly the symplectic rotations
\begin{equation}
U_{t}=%
\begin{pmatrix}
X_{t} & Y_{t}\\
-Y_{t} & X_{t}%
\end{pmatrix}
\label{utxybis}%
\end{equation}
are given by
\begin{equation}
X_{t}=(A_{t}A_{t}^{T}+B_{t}B_{t}^{T})^{-1/2}A_{t}\ \ ,\text{\ }Y_{t}%
=(A_{t}A_{t}^{T}+B_{t}B_{t}^{T})^{-1/2}B_{t}. \label{xyab}%
\end{equation}
The families $(V_{-P_{t}})$, $(M_{L_{t}})$, and $(U_{t})$ are symplectic
isotopies in their own right; they correspond to Hamiltonians that we will
denote by $H_{V}$, $H_{L}$, $H_{U}$. It is easy to find explicit expressions
for these Hamiltonians, and to show that the Hamiltonian function (\ref{hst})
determining the symplectic isotopy $(S_{t})$ can be written as a sum of
Hamiltonians. This is what we call the \textquotedblleft Iwasawa
sum\textquotedblright:

\begin{proposition}
(i) The symplectic isotopies $(V_{-P_{t}})$, $(M_{L_{t}})$, and $(U_{t})$ are
the flows determined by the Hamiltonians:
\begin{equation}
H_{V}(z,t)=\frac{1}{2}\dot{P}_{t}x^{2}\text{ \ },\text{ \ }H_{L}(z,t)=-\dot
{L}_{t}L_{t}^{-1}x\cdot p\label{hvlu}%
\end{equation}
and
\begin{equation}
H_{U}(z,t)=\tfrac{1}{2}(\dot{Y}_{t}X_{t}^{T}-\dot{X}_{t}Y_{t}^{T})x^{2}%
+\tfrac{1}{2}(\dot{Y}_{t}X_{t}^{T}-\dot{X}_{t}Y_{t}^{T})p^{2}.\label{HU}%
\end{equation}
(ii) The Hamiltonian function $H$ can be written%
\begin{equation}
H(z,t)=H_{V}(z,t)+H_{L}(V_{P_{t}}z,t)+H_{U}(M_{L_{t}^{-1}}V_{P_{t}%
}z,t)\label{iwasum}%
\end{equation}
(iii) We also have
\begin{equation}
H(z,t)=H_{R}(z,t)+H_{U}(R_{t}^{-1}z,t)\label{hru}%
\end{equation}
where
\begin{equation}
H_{R}(z,t)=\frac{1}{2}(\dot{L}_{t}Q_{t}^{T}-\dot{Q}_{t}L_{t})x^{2}-\dot{L}%
_{t}L_{t}^{-1}p\cdot x.\label{hoho}%
\end{equation}

\end{proposition}

\begin{proof}
\textit{(i)} The formulas (\ref{hvlu}) immediately follow from (\ref{hst})
replacing $(S_{t})$ with $(V_{-P_{t}})$ and $(M_{L_{t}})$, respectively.
Notice that $\dot{L}_{t}L_{t}^{-1}$ is symmetric since $L_{t}>0$ is: $\dot
{L}_{t}L_{t}^{-1}=\frac{d}{dt}\operatorname{Log}L_{t}$. Formula (\ref{HU}) is
obtained by writing
\[
H_{U}(z,t)=-\frac{1}{2}J\dot{U}_{t}U_{t}^{-1}z^{2}%
\]
(see Corollary \ref{coro3}). \textit{(ii)} In view of formula (\ref{hkl}) in
Proposition \ref{propham} $(S_{t})=(V_{-P_{t}}M_{L_{t}}U_{t})$ is the flow
determined by $H=H_{V}\#H_{L}\#H_{U}$ hence (\ref{iwasum}). \textit{(iii)}
Formula (\ref{hru}) is obtained in a similar fashion writing $(S_{t}%
)=(R_{t}U_{t})$ and using the equality $H=H_{R}\#H_{U}$. Formula (\ref{hoho})
follows from the equality $H_{R}=-\frac{1}{2}J\dot{R}_{t}R_{t}^{-1}$ and using
formula (\ref{rtbis}).
\end{proof}

Notice that formula (\ref{hoho}) can be rewritten%
\begin{equation}
H_{R}=\frac{1}{2}N(t)x^{2}-\dot{L}_{t}L_{t}^{-1}p\cdot x \label{honk}%
\end{equation}
where the symmetric matrix $N(t)$ is given by
\begin{equation}
N(t)=\dot{L}_{t}L_{t}^{-1}P_{t}+P_{t}\dot{L}_{t}L_{t}^{-1}-\dot{P}_{t}.
\label{klp}%
\end{equation}
Notice that, conversely, every Hamiltonian of the type%
\[
H_{0}=\frac{1}{2}N(t)x^{2}-K(t)p\cdot x
\]
with $N(t)=N(t)^{T}$ and $K(t)=K(t)^{T}$ leads to a flow in
$\operatorname*{Sp}_{0}(n)$: the corresponding Hamilton equations are%
\begin{align*}
\dot{x}(t)  &  =-K(t)x(t)\\
\ \dot{p}(t)  &  =-N(t)x(t)+K(t)p(t)
\end{align*}
and the corresponding symplectic isotopy is of the type (\ref{rt}) since the
equation $\dot{x}(t)=-K(t)x(t)$ contains no term $p(t)$.

\section{Chalkboard Motions and their Shadows}

We now apply the notions developed in the previous sections to what we call
\textquotedblleft chalkboard motion\textquotedblright\ in phase space, and
thereafter study the orthogonal projections (or \textquotedblleft
shadows\textquotedblright) of these motions on subspaces. Chalkboard motion
essentially consists in moving and distorting an ellipsoid in phase space
while preserving its symplectic capacity (or area in the case $n=1$).

\subsection{The action of $\operatorname*{ISp}\nolimits_{0}(n)$ on symplectic
balls}

From now on the real number $\varepsilon>0$ has the vocation to be a small
radius. For instance, in our applications to quantum mechanics we will choose
$\varepsilon=\sqrt{\hbar}$; the need for \textquotedblleft
smallness\textquotedblright\ in classical considerations will actually only be
necessary in Section \ref{secnearby} where we study nonlinear evolution. We
call \textit{symplectic ball} the image of a ball $B^{2n}(\varepsilon)$ by
some element of $\operatorname*{ISp}(n)$. A symplectic ball can always be
written
\[
B_{S}^{2n}(z_{0},\varepsilon)=T(z_{0})SB^{2n}(\varepsilon)
\]
for some $S\in\operatorname*{Sp}(n)$ and $z_{0}$ will be called the center of
$B_{S}^{2n}(z_{0},\varepsilon)$. Let $\operatorname*{SBall}_{\varepsilon}(2n)$
be the set of all symplectic balls in $(\mathbb{R}^{2n},\varepsilon)$ with the
same radius $\varepsilon$ (equivalently, with the same symplectic capacity
$\pi\varepsilon^{2}$); we have a natural transitive action
\[
\operatorname*{ISp}(n)\times\operatorname*{SBall}\nolimits_{\varepsilon
}(2n)\longrightarrow\operatorname*{SBall}\nolimits_{\varepsilon}(2n).
\]
It turns out that the restriction
\[
\operatorname*{ISp}\nolimits_{0}(n)\times\operatorname*{SBall}%
\nolimits_{\varepsilon}(2n)\longrightarrow\operatorname*{SBall}%
\nolimits_{\varepsilon}(2n)
\]
of this action to the local inhomogeneous symplectic group
$\operatorname*{ISp}\nolimits_{0}(n)$ is also transitive:

\begin{proposition}
\label{propsplp}(i) Every symplectic ball $B_{S}^{2n}(z_{0},\varepsilon)$ can
be obtained from the ball $B^{2n}(\varepsilon)$ using the local subgroup
$\operatorname*{ISp}_{0}(n)$ of $\operatorname*{ISp}(n)$. In fact, for every
$S\in\operatorname*{Sp}(n)$ there exist unique $P=P^{T}$, $L=L^{T}$, and
$z_{0}\in\mathbb{R}^{2n}$ such that \
\begin{equation}
B_{S}^{2n}(z_{0},\varepsilon)=T(z_{0})V_{P}M_{L}B^{2n}(\varepsilon).
\label{bs}%
\end{equation}
(ii) More generally, if $S=V_{P}M_{L}$ and $S^{\prime}=V_{P^{\prime}%
}M_{L^{\prime}}$ then%
\begin{equation}
B_{S^{\prime}}^{2n}(z_{0}^{\prime},\varepsilon)=S(P,L,P^{\prime},L^{\prime
},z_{0},z_{0}^{\prime})B_{S}^{2n}(z_{0},\varepsilon) \label{A10}%
\end{equation}
with $S(P,L,P^{\prime},L^{\prime},z_{0},z_{0}^{\prime})\in\operatorname*{ISp}%
_{0}(n)$ given by%
\begin{equation}
S(P,L,P^{\prime},L^{\prime},z_{0},z_{0}^{\prime})=T(z_{0}^{\prime}-Rz_{0})R
\label{B10}%
\end{equation}
where $R\in\operatorname*{ISp}_{0}(n)$ is the product
\begin{equation}
R=V_{P^{\prime}-(L^{\prime}L^{-1})^{T}PL^{-1}L^{\prime}}M_{L^{-1}L^{\prime}}.
\label{C11}%
\end{equation}

\end{proposition}

\begin{proof}
\textit{(i)} Using a pre-Iwasawa\ factorization we can find unique $P=P^{T}$,
$L=L^{T}>0$ given by (\ref{pl1}), (\ref{pl2}), and a symplectic rotation
$U_{X,Y}\in U(n)$ such that $S=V_{P}M_{L}U_{X,Y}$; formula (\ref{bs}) follows
since by rotational symmetry, $U_{X,Y}B^{2n}(\varepsilon)=B^{2n}(\varepsilon)$
and thus%
\begin{align*}
Q_{S}^{2n}(z_{0})  &  =T(z_{0})V_{P}M_{L}U_{X,Y}B^{2n}(\varepsilon)\\
&  =T(z_{0})V_{P}M_{L}B^{2n}(\varepsilon).
\end{align*}
The uniqueness of a transformation $T(z_{0})V_{P}M_{L}\in\operatorname*{ISp}%
_{0}(n)$ such that (\ref{bs}) holds is easily verified: suppose that
\[
T(z_{0})V_{P}M_{L}B^{2n}(\varepsilon)=T(z_{0}^{\prime})V_{P^{\prime}%
}M_{L^{\prime}}B^{2n}(\varepsilon)
\]
then there exists $U\in U(n)$ such that
\[
T(z_{0})V_{P}M_{L}U=T(z_{0}^{\prime})V_{P^{\prime}}M_{L^{\prime}}.
\]
This implies that we must have $z_{0}=z_{0}^{\prime}$, applying both sides to
$z=0$. But then $V_{P}M_{L}U=V_{P^{\prime}}M_{L^{\prime}}$ which is only
possible if $U$ is the identity, so we have $V_{P}M_{L}=V_{P^{\prime}%
}M_{L^{\prime}}$ which implies $P=P^{\prime}$ and $L=L^{\prime}$.
\textit{(ii)} Let $B_{S}^{2n}(\varepsilon)$ $=SB^{2n}(\varepsilon)$ and
$B_{S^{\prime}}^{2n}(\varepsilon)$ $=S^{\prime}B^{2n}(\varepsilon)$ be two
symplectic balls centered at $0$; we thus have $B_{S^{\prime}}^{2n}%
(\varepsilon)=S^{\prime}S^{-1}B_{S}^{2n}(\varepsilon)$. Taking $S=V_{P}M_{L}$
and $S^{\prime}=V_{P^{\prime}}M_{L^{\prime}}$ we have, in view of formula
(\ref{xinvy}),
\[
S^{\prime}S^{-1}=V_{P^{\prime}-(L^{\prime}L^{-1})^{T}PL^{-1}L^{\prime}%
}M_{L^{-1}L^{\prime}}%
\]
proving (\ref{A10}) for $z_{0}=z_{0}^{\prime}=0$. The case of arbitrary
centers $z_{0}$ and $z_{0}^{\prime}$ $\frac{{}}{{}}$readily follows: assume
that $B_{S}^{2n}(z_{0},\varepsilon)$ and $B_{S^{\prime}}^{2n}(z_{0}^{\prime
},\varepsilon)$ are centered at $z_{0}$ and $z_{0}^{\prime}$, respectively. We
have $B_{S}^{2n}(z_{0},\varepsilon)=T(z_{0})SB^{2n}(\varepsilon)$ hence
\begin{align*}
B_{S^{\prime}}^{2n}(z_{0}^{\prime},\varepsilon)  &  =T(z_{0}^{\prime
})S^{\prime}(T(z_{0})S)^{-1}B_{S}^{2n}(z_{0},\varepsilon)\\
&  =T(z_{0}^{\prime}-S^{\prime}S^{-1}z_{0})S^{\prime}S^{-1}B_{S}^{2n}%
(z_{0},\varepsilon).
\end{align*}
Choosing $S=V_{P}M_{L}$ and $S^{\prime}=V_{P^{\prime}}M_{L^{\prime}}$ as above
we are done.
\end{proof}

\subsection{Linear and affine chalkboard motions\label{secaff}}

It will be convenient to write (\ref{rt}) as above in the form%
\begin{equation}
R_{t}=%
\begin{pmatrix}
L_{t}^{-1} & 0\\
Q_{t} & L_{t}%
\end{pmatrix}
\text{ \ \textit{with} \ }Q_{t}=P_{t}L_{t}^{-1}. \label{rtbis}%
\end{equation}

Recall that the symplectic isotopy $(R_{t})$ is determined by a Hamiltonian
which does not contain any kinetic term.

Let $(T(z_{t})S_{t})$ be a chalkboard motion; then%
\[
T(z_{t})S_{t}B^{2n}(\varepsilon)=T(z_{t})R_{t}B^{2n}(\varepsilon)
\]
which shows that this motion is entirely determined by the trajectory of the
center of $B^{2n}(\varepsilon)$ and the flow determined by a reduced
Hamiltonian of the type (\ref{hoho})--(\ref{honk}). This immediately follows
from the Iwasawa factorization $S_{t}=R_{t}U_{t}$ since balls centered at the
origin are invariant under the action of the group $U(n)$.

More generally:

\begin{proposition}
\label{propchalk1}Let $B_{S}^{2n}(a,\varepsilon)=T(a)SB^{2n}(\varepsilon)$ be
a symplectic ball. We have
\begin{equation}
(T(z_{t})S_{t})B_{S}^{2n}(a,\varepsilon)=B_{R_{t}S}^{2n}(a_{t},\varepsilon)
\label{omegat1}%
\end{equation}
where $a_{t}=S_{t}a$ and $(R_{t})$ is a symplectic isotopy in
$\operatorname*{Sp}_{0}(n)$ defined as follows: let $H$ be the quadratic
Hamiltonian generating $(S_{t})$, that is $H=-\frac{1}{2}J\dot{S}_{t}%
S_{t}^{-1}z^{2}$. Then $(R_{t})$ is the reduced flow determined by $H\circ R$
where $R$ is the local part in the pre-Iwasawa factorization $S=RU$.
\end{proposition}

\begin{proof}
Let $S=RU$ be the pre-Iwasawa factorization of $S$. Since $S_{t}%
T(a)=T(S_{t}a)S_{t}$ and $UB^{2n}(\varepsilon)=B^{2n}(\varepsilon)$ we have
\begin{align*}
(T(z_{t})S_{t})B_{S}^{2n}(a,\varepsilon)  &  =T(z_{t})S_{t}T(a)RUB^{2n}%
(\varepsilon)\\
&  =T(z_{t}+a_{t})S_{t}RB^{2n}(\varepsilon)\\
&  =T(z_{t}+a_{t})R(R^{-1}S_{t}R)B^{2n}(\varepsilon).
\end{align*}
In view of the conjugation formula (\ref{conj}), $(S_{t}^{\prime}%
)=(R^{-1}S_{t}R)$ is the symplectic isotopy generated by the quadratic
Hamiltonian $H\circ S$; the latter is given by%
\[
H\circ R(z,t)=-\frac{1}{2}JR^{-1}\dot{S}_{t}S_{t}^{-1}Rz^{2}.
\]
We now apply the pre-Iwasawa factorization to $S_{t}^{\prime}$ and write
$S_{t}^{\prime}=R_{t}^{\prime}U_{t}^{\prime}$ so that
\begin{align*}
(T(z_{t})S_{t})B_{S}^{2n}(a,\varepsilon)  &  =T(z_{t}+a_{t})RR_{t}^{\prime
}U_{t}^{\prime}B^{2n}(\varepsilon)\\
&  =T(z_{t}+a_{t})RR_{t}^{\prime}B^{2n}(\varepsilon)\\
&  =T(z_{t}+a_{t})(RR_{t}^{\prime}R^{-1})RB^{2n}(\varepsilon)\\
&  =T(z_{t}+a_{t})(RR_{t}^{\prime}R^{-1})SB^{2n}(\varepsilon).
\end{align*}
hence the equality (\ref{omegat1}) with $R_{t}=RR_{t}^{\prime}R^{-1}$.
\end{proof}

\subsection{The nonlinear case: nearby orbit approximation\label{secnearby}}

Assume now that we have a method allowing us to displace and deform any
ellipsoid $\Omega=T(z_{0})FB^{2n}(\varepsilon)$, $F\in GL(2n,\mathbb{R})$, in
such a way that its symplectic capacity remains constant. More explicitly, we
make the assumption that there exists a $C^{1}$ curve $t\longmapsto z_{t}$
starting from $z_{0}$ and a family $(g_{t})$ of $C^{1}$ diffeomorphisms
satisfying $g_{0}=I_{\mathrm{d}}$ together with the equilibrium condition%
\begin{equation}
g_{t}(0)=0\text{ \ \textit{for} \ }t\in I_{T}. \label{equil}%
\end{equation}
At time $t\in I_{T}$ the ellipsoid $\Omega$ becomes a (usually not elliptic)
set
\[
\Omega_{t}=T(z_{t})g_{t}T(-z_{0})\Omega
\]
such that $c(\Omega_{t})=c(\Omega)$; we thus have $\Omega_{t}=f_{t}(\Omega)$
where the symplectomorphisms $f_{t}$ are defined by
\begin{equation}
f_{t}=T(z_{t})g_{t}T(-z_{0})=T(z_{t}-z_{0})T(z_{0})g_{t}T(-z_{0}) \label{fto}%
\end{equation}
and it follows from Proposition \ref{thmchalk} that this motion must be
Hamiltonian: $(f_{t})$ is a\ symplectic isotopy generated by a Hamiltonian
function $H$ which we determine now.

\begin{proposition}
\label{propchalk3}The symplectic isotopy $(f_{t})$ defined by
\begin{equation}
f_{t}=T(z_{t})g_{t}T(-z_{0}) \label{ftg}%
\end{equation}
is the Hamiltonian flow determined by the Hamiltonian%
\begin{equation}
H(z,t)=H_{2}(z-z_{t},t)+\sigma(z,\dot{z}_{t}) \label{hsig}%
\end{equation}
where
\begin{equation}
H_{2}(z,t)=-\int_{0}^{1}\sigma(\dot{g}_{t}g_{t}^{-1}(\lambda z),z)d\lambda
\label{ha2}%
\end{equation}
is the Hamiltonian function generating $(g_{t})$.
\end{proposition}

\begin{proof}
Let us write
\begin{equation}
f_{t}=T(z_{t}-z_{0})T(z_{0})g_{t}T(-z_{0}).
\end{equation}
We first remark that $t\longmapsto T(z_{t}-z_{0})$ is the flow determined by
the translation Hamiltonian $H_{1}(z,t)=\sigma(z,\dot{z}_{t})$. The symplectic
isotopy $(g_{t})$ is determined by (\ref{ha2}) (Proposition \ref{thm1}), and
in view of the conjugation property (\ref{conj}) the flow $t\longmapsto
T(z_{0})g_{t}T(-z_{0})$ is thus determined by $H_{2}\circ T(-z_{0})$. Formula
(\ref{hsig}) now follows from the product property (\ref{ch1}) of Hamiltonian flows.
\end{proof}

Assuming that the radius $\varepsilon$ is small it makes sense to replace the
symplectic isotopy $(g_{t})$ with its linearization $(g_{t}^{0})$ around its
equilibrium point $z=0$ (Arnol'd \cite{Arnold}, \S 5.22), that is, we take
\begin{equation}
g_{t}^{0}(z)=g_{t}(0)+Dg_{t}(0)z=S_{t}^{0}z \label{gto}%
\end{equation}
where $S_{t}^{0}=Dg_{t}(0)$ is the Jacobian matrix of $g_{t}$ calculated at
the origin. A classical result (see \textit{e.g.} \cite{Birk}, \S 2.3.2) tells
us that $t\longmapsto S_{t}^{0}$ satisfies the \textquotedblleft variational
equation\textquotedblright\
\[
\frac{d}{dt}S_{t}^{0}=JD_{z}^{2}H(g_{t}(0),t)S_{t}^{0}=JD_{z}^{2}%
H(0,t)S_{t}^{0}%
\]
and hence $(g_{t}^{0})=(S_{t}^{0})$ is the flow determined by the quadratic
Hamiltonian function%
\begin{equation}
H_{2}^{0}(z,t)=\frac{1}{2}D_{z}^{2}H(0,t)z^{2}. \label{hotwo}%
\end{equation}
With this approximation the symplectomorphisms (\ref{ftg}) are replaced with%
\begin{equation}
f_{t}^{0}=T(z_{t}-z_{0})T(z_{0})S_{t}^{0}T(-z_{0}) \label{ftobis}%
\end{equation}
and $(f_{t}^{0})$ is the flow determined by the Hamiltonian%
\begin{equation}
H^{0}(z,t)=\frac{1}{2}D_{z}^{2}H(0,t)(z-z_{t})^{2}+\sigma(z,\dot{z}_{t}).
\label{ho}%
\end{equation}

We remark that $H_{2}^{0}(z,t)$ is obtained from the \textquotedblleft
exact\textquotedblright\ Hamiltonian (\ref{ha2}) by truncating the Taylor
series of $H_{2}$ at $z=0$ by dropping third order terms and above: noting
that $\partial_{z}H(0,t)=0$ since $0$ is an equilibrium point we have
\begin{equation}
H_{2}(z,t)=H_{2}(0,t)+\tfrac{1}{2}D_{z}^{2}H_{2}(0,t)z^{2}+O(z^{3}) \label{k}%
\end{equation}
and the term $H_{2}(0,t)$ can be neglected. Similarly, dismissing the terms
$H_{2}(0,t)$ and $\sigma(z_{t},\dot{z}_{t})$,
\begin{equation}
H(z,t)=\tfrac{1}{2}D_{z}^{2}H_{2}(0,t)(z-z_{t})^{2}+\sigma(z-z_{t},\dot{z}%
_{t})+O((z-z_{t})^{3}) \label{kbis}%
\end{equation}
so our method is closely related to the so-called \textquotedblleft nearby
orbit method\textquotedblright\ popularized by researchers working in
semiclassical approximations \cite{coro97,heller,huheli,Littlejohn}. In this
method one expands the Hamiltonian around an orbit and truncates the Taylor
series in order to get a more tractable problem.

A natural question arises at this point: in view of (\ref{kbis}) we have
$f_{t}^{H}(z_{0})=f_{t}^{H_{0}}(z_{0})=z_{t}$. What can we say about the
difference $f_{t}^{H}(z_{1})-f_{t}^{H_{0}}(z_{1})$ for an arbitrary point
$z_{1}\in\Omega$? Intuitively, the smaller the radius $\varepsilon$ is, the
better will $f_{t}^{H}(z_{1})$ approximate $f_{t}^{H}(z_{0})$ (at least for
not too big times $t$). Let us briefly discuss this without going too much
into the theory of the stability of Hamiltonian systems, for which there
exists an immense literature. See the recent preprint by Hong Qin \cite{hong}
for new results in the case of periodic orbits.

\subsection{A recalibration procedure}

We now briefly discuss an option which, to the best of our knowledge, has not
been explored yet, and leads to open question. Recall that we discussed in the
Introduction the John--L\"{o}wner ellipsoid. As before we start with a
symplectic ball $B_{S}^{2n}(\varepsilon)$, which we suppose centered at the
origin for simplicity. We displace $B_{S}^{2n}(\varepsilon)$ along a curve
$(z_{t})$ in phase space while deforming it using a symplectic isotopy
$(g_{t})$ consisting of an arbitrary family of symplectomorphisms starting
from the identity at time $t=0$. If we assume that this deformation is
sufficiently \textquotedblleft gentle\textquotedblright\ and preserves the
convexity, a natural idea is to replace $\Omega_{t}=g_{t}(B_{S}^{2n}%
(\varepsilon))$ with an ellipsoid $\widetilde{\Omega}_{t}$ having the same
symplectic capacity $\pi\varepsilon^{2}$ as $\Omega_{t}$: $c(\widetilde{\Omega
}_{t})=c(\widetilde{\Omega}_{t})$ for some choice of symplectic capacity $c$.
It turns out that there exists a unique maximum volume ellipsoid containing
$\Omega_{t}$, and by dilation one can obtain a minimum volume ellipsoid
containing $g_{t}(B_{S}^{2n}(\varepsilon))$. Now, volume is not related to
symplectic capacity (except in the case $n=1$ where both notions coincide with
area for connected simply connected surfaces), and it is not known whether one
can construct a \textquotedblleft minimum (or maximum) capacity
ellipsoid\textquotedblright\ which is the symplectic analogue of the
John--L\"{o}wner ellipsoid. However, we can do the following
\cite{Ball,turc1,turc2}. Among the ellipsoids circumscribing $\Omega_{t}$,
there exists a unique one $\Omega_{t}^{\min}$ with minimum volume (the
L\"{o}wner ellipsoid) and similarly, among the ellipsoids inscribed in
$\Omega_{t}$, there exists a unique one $\Omega_{t}^{\max}$ of maximum volume
(the John ellipsoid) and we have
\[
\frac{1}{n}\Omega_{t}^{\min}\subset\widetilde{\Omega}_{t}\subset\Omega
_{t}^{\min}\text{ \ , \ }\Omega_{t}^{\max}\subset\widetilde{\Omega}_{t}\subset
n\Omega_{t}^{\max}.
\]
In case $\Omega_{t}$ is symmetric (\textit{i.e.} $\Omega_{t}=-\Omega_{t}$) the
coefficients $1/n$ and $n$ can be changed into $1/\sqrt{n}$ and $\sqrt{n}$.
Also note that this also works when $\Omega_{t}$ fails to be a convex body. It
suffices to replace $\Omega_{t}$ with its convex hull.

\subsection{The shadow of a chalkboard motion}

We now study the projection (\textquotedblleft shadow\textquotedblright) of
the chalkboard motion on the \textquotedblleft configuration
space\textquotedblright\ $X=\mathbb{R}_{x}^{n}$. For this the following lemma
about projections of ellipsoids will be helpful:

\begin{lemma}
\label{lemmapim}Let $\Omega=\{Mz^{2}\leq\varepsilon^{2};z\in\mathbb{R}^{2n}\}$
and assume that $M=M^{T}>0$ is given in $n\times n$ block form by%
\begin{equation}
M=%
\begin{pmatrix}
M_{XX} & M_{XP}\\
M_{PX} & M_{PP}%
\end{pmatrix}
\label{mxp}%
\end{equation}
(thus $M_{XP}=M_{PX}^{T}$). Let $\Pi$ be the orthogonal projection
$\mathbb{R}^{2n}\longrightarrow\mathbb{R}_{x}^{n}\times\{0\}$. We have
\begin{equation}
\Pi\Omega=\{x\in\mathbb{R}^{n}:(M/M_{PP})x^{2}\leq\varepsilon^{2}\}
\label{pambb}%
\end{equation}
where the $n\times n$ matrix
\begin{equation}
M/M_{PP}=M_{XX}-M_{XP}M_{PP}^{-1}M_{PX} \label{Schur}%
\end{equation}
is the Schur complement of the block $M_{PP}$ of $M$.
\end{lemma}

\begin{proof}
Recall \cite{zhang} that if $M>0$ then $M/M_{PP}>0$ and hence the ellipsoid
(\ref{pambb}) is nondegenerate. Set $Q(z)=Mz^{2}-\varepsilon^{2}$; the
hypersurface $\partial\Omega:Q(z)=0$ bounding $\Omega$ is defined by%
\begin{equation}
M_{XX}x^{2}+2M_{PX}x\cdot p+M_{PP}p^{2}=\varepsilon^{2}. \label{mabab}%
\end{equation}
The normal vectors to the boundary of $\Pi_{X}\Omega$ must stay in $X$ hence
the constraint $\partial_{z}Q(z)=2Mz\in\mathbb{R}^{n}\times\{0\}$, which is
equivalent to $M_{PX}x+M_{PP}p=0$, that is to $p=-M_{PP}^{-1}M_{PX}x$.
Inserting $p$ in (\ref{mabab}) shows that $\Pi_{X}\Omega$ is bounded by
$\Sigma_{X}:(M/M_{PP})x^{2}=\varepsilon^{2}$ which yields (\ref{pambb}).
\end{proof}

This result easily allows us to find the orthogonal projection of a symplectic
ball on the $x$-space $\mathbb{R}^{n}\times\{0\}$ (or on the $p$-space
$\{0\}\times\mathbb{R}^{n}$); we will more generally consider the projection
of a chalkboard motion:

\begin{proposition}
Let $(T(z_{t})S_{t})$ be a symplectic isotopy with%
\begin{equation}
S_{t}=%
\begin{pmatrix}
A_{t} & B_{t}\\
C_{t} & D_{t}%
\end{pmatrix}
\text{ \ , \ }S_{0}=I_{\mathrm{d}}.
\end{equation}
The orthogonal projection of the symplectic ball
\[
B_{S_{t}}^{2n}(z_{t},\varepsilon)=T(z_{t})S_{t}B^{2n}(\varepsilon)
\]
on the configuration space $\mathbb{R}^{n}\times\{0\}$ is the ellipsoid
\begin{equation}
\Pi B_{S_{t}}^{2n}(z_{t},\varepsilon)=T(x_{t},0)(A_{t}A_{t}^{T}+B_{t}B_{t}%
^{T})^{1/2}B^{n}(\varepsilon). \label{projchalk}%
\end{equation}

\end{proposition}

\begin{proof}
It is no restriction to assume $z_{t}=0$ since translations project to
translations in the first component. We have $S_{t}B^{2n}(\varepsilon
)=R_{t}B^{2n}(\varepsilon)$ where $(R_{t})$ is the symplectic isotopy in
$\operatorname*{ISp}_{0}(n)$ given by (\ref{rtbis}), that is
\begin{equation}
R_{t}=%
\begin{pmatrix}
L_{t}^{-1} & 0\\
Q_{t} & L_{t}%
\end{pmatrix}
,
\end{equation}
the matrices $Q_{t}$ and $L_{t}$ being given by
\begin{align}
Q_{t}  &  =(C_{t}A_{t}^{T}+D_{t}B_{t}^{T})(A_{t}A_{t}^{T}+B_{t}B_{t}%
^{T})^{-1/2}\\
L_{t}  &  =(A_{t}A_{t}^{T}+B_{t}B_{t}^{T})^{-1/2}.
\end{align}
The ellipsoid $R_{t}B^{2n}(\varepsilon)$ is the set of all $z\in
\mathbb{R}^{2n}$ such that $(R_{t}R_{t}^{T})^{-1}z^{2}\leq\varepsilon$ hence
the matrix $M$ in (\ref{mxp}) is given by%
\[
M=%
\begin{pmatrix}
QQ^{T}+L^{2} & -QL^{-1}\\
-L^{-1}Q^{T} & L^{-2}%
\end{pmatrix}
\]
and the Schur complement $M/M_{PP}$ is then just $L_{t}^{2}=(A_{t}A_{t}%
^{T}+B_{t}B_{t}^{T})^{-1}$. Formula (\ref{projchalk}) follows.
\end{proof}

Formula (\ref{projchalk}) perfectly illustrates that the phenomenon of
\textquotedblleft spreading\textquotedblright\ is not, \textit{per se}, a
quantum phenomenon as many physicists still believe (this was in fact remarked
on a long time ago by Littlejohn \cite{Littlejohn}). In fact spreading will
always occur provided that $A_{t}A_{t}^{T}+B_{t}B_{t}^{T}$ is not constant,
that is equal to $I_{\mathrm{d}}$ for all $t$. For instance, in the case $n=1$
we would have $A_{t}^{2}+B_{t}^{2}=1$ and the phase space motion would be a
rotation leaving the disk $|z|\leq\varepsilon$ invariant.

Let us next briefly consider the case of subsystems, obtained by orthogonal
projection on a smaller phase space. In the quantum case such projections
intervene in the study of entanglement. Let again $\Omega=FB^{2n}%
(\varepsilon)$ be a non-degenerate ellipsoid centered at the origin; setting
$M=(FF^{T})^{-1}$\ this ellipsoid is the set of all $z\in\mathbb{R}^{2n}$ such
that $Mz^{2}\leq R^{2}$; $M$ is a symmetric and positive definite matrix.
Consider now the splitting $\mathbb{R}^{2n_{A}}\oplus\mathbb{R}^{2n_{B}}$ of
$\mathbb{R}^{2n}$ with $n_{A}+n_{B}=n$. The spaces $\mathbb{R}^{2n_{A}}%
\equiv\mathbb{R}_{x_{A}}^{n_{A}}\times\mathbb{R}_{p_{A}}^{n_{A}}$ and
$\mathbb{R}^{2n_{B}}\equiv\mathbb{R}_{x_{B}}^{n_{B}}\times\mathbb{R}_{p_{_{B}%
}}^{n_{B}}$ are viewed as the phase spaces of two subsystems $A$ and $B$. We
will write the matrix of $M$ as
\begin{equation}
M=%
\begin{pmatrix}
M_{AA} & M_{AB}\\
M_{BA} & M_{BB}%
\end{pmatrix}
\label{MMM}%
\end{equation}
where the blocks $M_{AA}$, $M_{AB}$, $M_{BA}$, $M_{BB}$ have, respectively,
dimensions $n_{A}\times n_{A}$, $n_{A}\times n_{B}$, $n_{B}\times n_{A}$,
$n_{B}\times n_{B}$. Since $M$ is positive definite and symmetric (written
$M>0$) the blocks $M_{AA}$ and $M_{BB}$ are symmetric and positive definite
(and hence invertible) and $M_{BA}=M_{AB}^{T}$.

Abbondandolo and Matveyev \cite{abbo} (also see \cite{abbobis}) have shown
that for $S\in\operatorname*{Sp}(n)$.%
\begin{equation}
\operatorname*{Vol}\nolimits_{2n_{A}}\Pi_{A}S(B^{2n}(R)\geq\frac{(\pi
R^{2})^{n_{A}}}{n_{A}!}\label{abbo1}%
\end{equation}
for every $R>0$. This inequality can be seen as an interpolation between
Gromov's theorem and Liouville's theorem on the conservation of volume for
$2\leq n_{A}\leq n-1$. In \cite{DiGoPra19} we have proven the following
improvement of Abbondandolo and Matveyev's result: 

\begin{proposition}
\label{PropAB}Let $\Pi_{A}$ be the orthogonal projection $\mathbb{R}^{2n_{A}%
}\times\mathbb{R}^{2n_{B}}\longrightarrow\mathbb{R}^{2n_{A}}$. There exists
$S_{A}\in\operatorname*{Sp}(n_{A})$ such that for every $R>0$ the projected
ellipsoid $\Pi_{A}(S(B_{R}^{2n}))$ contains the symplectic ball $S_{A}%
(B_{R}^{2n_{A}})$:%
\begin{equation}
\Pi_{A}(S(B_{R}^{2n}))\supset S_{A}(B_{R}^{2n_{A}}))~.\label{projo}%
\end{equation}
More generally,
\begin{equation}
\Pi_{A}(S(B_{R}^{2n}(z_{0})))\supset\{\Pi_{A}(Sz_{0})\}+S_{A}(B_{R}^{2n_{A}%
})~.\label{proja}%
\end{equation}
We have equality in (\ref{projo}), (\ref{proja}) if and only if $S=S_{A}\oplus
S_{B}$ for some $S_{B}\in\operatorname*{Sp}(n_{B})$.
\end{proposition}

The proof of this result makes use of the properties of the projection
operator using Williamson's symplectic diagonalization theorem and is thus
conceptually somewhat simpler than the method of Abbondandolo and Matveyev
which uses the Wirtinger inequality from the theory of differential forms. We
mention that Abbondandolo and Benedetti \cite{abbonew} have very recently
improved and extended the scope of their results  using the theory of Zoll
contact forms. 

\section{Quantum Blobs and the Wigner Transform}

The symplectic group and its double covering, the metaplectic group, are the
keystones of mechanics in both their Hamiltonian and quantum formulations.
Loosely speaking one can say that the passage from the symplectic group to its
metaplectic representation is the shortest bridge between classical and
quantum mechanics \cite{RMP}, because it provides us automatically with the
Weyl quantization of quadratic Hamiltonians without any recourse to physical
arguments (also see the discussion in \cite{gohi1}).

\subsection{The local metaplectic group\label{secmet}}

The symplectic group $\operatorname*{Sp}(n)$ is a connected matrix Lie group,
contractible to its compact subgroup $U(n)$ and has covering groups of all
orders. The metaplectic group $\operatorname*{Mp}(n)$ is a unitary
representation in $L^{2}(\mathbb{R}^{n})$ of the double cover of
$\operatorname*{Sp}(n)$ by unitary operators acting on square integrable
functions (see \cite{Birk}, Chapter 7 for a detailed study and construction of
$\operatorname*{Mp}(n)$). To every $S\in\operatorname*{Sp}(n)$ the metaplectic
representation associates two unitary operators $\pm\widehat{S}\in
\operatorname*{Mp}(n)$ on $L^{2}(\mathbb{R}^{n})$. This representation is
entirely determined by its action on the generators of $\operatorname*{Mp}(n)$
since the covering projection $\pi^{\operatorname*{Mp}}:\operatorname*{Mp}%
(n)\longrightarrow\operatorname*{Sp}(n)$ is a group epimorphism. This
correspondence is summarized in the table below:

\begin{center}%
\begin{tabular}
[c]{|l|l|l|}\hline
$\widehat{J}\psi(x)=\left(  \tfrac{1}{2\pi i\hbar}\right)  ^{n/2}\int
e^{-\frac{1}{\hbar}x\cdot x^{\prime}}\psi(x^{\prime})d^{n}x^{\prime}$ &
$\overset{\pi^{\operatorname*{Mp}}}{\longrightarrow}$ & $J=%
\begin{pmatrix}
0 & I\\
-I & 0
\end{pmatrix}
$\\\hline
$\widehat{V}_{P}\psi(x)=e^{-\frac{i}{2\hbar}Px\cdot x}\psi(x)$ &
$\overset{\pi^{\operatorname*{Mp}}}{\longrightarrow}$ & $V_{P}=%
\begin{pmatrix}
I & 0\\
-P & I
\end{pmatrix}
$\\\hline
$\widehat{M}_{L,m}\psi(x)=i^{m}\sqrt{|\det L|}\psi(Lx)$ & $\overset{\pi
^{\operatorname*{Mp}}}{\longrightarrow}$ & $M_{L}=%
\begin{pmatrix}
L^{-1} & 0\\
0 & L^{T}%
\end{pmatrix}
$\\\hline
\end{tabular}

\end{center}

\noindent the integer $m$ in $\widehat{M}_{L,m}$ (\textquotedblleft Maslov
index\textquotedblright) being chosen so that $\arg\det L=m\pi$
$(\operatorname{mod}2\pi)$. Observe that these operators (and hence every
$\widehat{S}\in\operatorname*{Mp}(n)$) can be extended into continuous
operators acting on the Schwartz space $\mathcal{S}^{\prime}(\mathbb{R}^{n})$
of tempered distributions. Needless to say there are other ways to introduce
the metaplectic group. A good way is to use generalized Fourier transforms and
the apparatus of generating functions (see \cite{Birk}, Chapter 7). It works
as follows: assume that
\[
S=%
\begin{pmatrix}
A & B\\
C & D
\end{pmatrix}
\text{ \ , \ }\det B\neq0
\]
is a free symplectic matrix; to $S$ we associate its generating function%
\[
W(x,x^{\prime})=\frac{1}{2}DB^{-1}x^{2}-B^{-1}x\cdot x^{\prime}+\frac{1}%
{2}B^{-1}Ax^{\prime2}.
\]
This function has the property that
\[
(x,p)=S(x^{\prime},p^{\prime})\Longleftrightarrow\left\{
\begin{array}
[c]{c}%
p=\partial_{x}W(x,x^{\prime})\\
p^{\prime}=-\partial_{x^{\prime}}W(x,x^{\prime})
\end{array}
\right.
\]
as can be checked by a direct calculation. Now, exactly every element of
$\operatorname*{Sp}(n)$ is the product of two free symplectic matrices, every
$\widehat{S}\in\operatorname*{Mp}(n)$ can be written (non uniquely) as a
product of two Fourier integral operators of the type $\widehat{S}_{W,m}$
where%
\begin{equation}
\widehat{S}_{W,m}\psi(x)=\left(  \tfrac{1}{2\pi i\hbar}\right)  ^{n/2}%
i^{m}\sqrt{|\det B^{-1}|}\int e^{\frac{i}{\hbar}W(x,x^{\prime})}\psi
(x^{\prime})d^{n}x^{\prime} \label{swm}%
\end{equation}
where $m$ is an integer $\operatorname{mod}4$ corresponding to a choice of
$\arg\det B^{-1}$. We notice that $\widehat{S}_{W,m}$ can be simply expressed
in terms of the unitary operators $\widehat{J}$, $\widehat{V}_{P}$ and
$\widehat{M}_{L,m}$ defined above: a simple inspection of the formula above
shows that (\textit{cf.} formula (\ref{sdba}) in Section \ref{seclocal})
\begin{equation}
\widehat{S}_{W,m}=\widehat{V}_{-DB^{-1}}\widehat{M}_{B^{-1},m}\widehat{J}%
\widehat{V}_{-B^{-1}A}. \label{swmbis}%
\end{equation}

We now define the \textit{local metaplectic group}: it is the subgroup
$\operatorname*{Mp}_{0}(n)$ of $\operatorname*{Mp}(n)$ generated by the
operators $\widehat{M}_{L,m}$ and $\widehat{V}_{P}$. This group actually
consists of all products $\widehat{V}_{P}\widehat{M}_{L,m}$ (or $\widehat{M}%
_{L,m}\widehat{V}_{P}$) as follows from the formulas (\textit{cf}.
(\ref{mlvp}) and (\ref{mlpinv})):
\begin{equation}
\widehat{M}_{L,m}\widehat{V}_{P}=\widehat{V}_{L^{T}PL}\widehat{M}_{L,m}\text{
\ },\text{ \ }\widehat{V}_{P}\widehat{M}_{L,m}=\widehat{M}_{L,m}%
\text{\ }\widehat{V}_{(L^{-1})^{T}PL^{-1}} \label{MLVP}%
\end{equation}
and, using the relations $(\widehat{V}_{P})^{-1}=\widehat{V}_{-P}$,
$(\widehat{M}_{L,m})^{-1}=\widehat{M}_{L^{-1},-m}$,
\begin{equation}
(\widehat{V}_{P}\widehat{M}_{L,m})^{-1}=\widehat{V}_{(L^{-1})^{T}PL^{-1}%
}\widehat{M}_{L^{-1},-m}. \label{mlvpinv}%
\end{equation}
Combining these relations we get the following formula, which is the
metaplectic analogue of (\ref{xinvy}):%
\begin{equation}
\widehat{S^{\prime}}\widehat{S}^{-1}=\widehat{V}_{P^{\prime}-(L^{-1}L^{\prime
})^{T}P(L^{-1}L^{\prime})}\widehat{M}_{L^{-1}L^{\prime},-m+m^{\prime}}.
\label{XINVY}%
\end{equation}

The local symplectic group $\operatorname*{Sp}_{0}(n)$ is the image of the
local metaplectic group $\operatorname*{Mp}_{0}(n)$ by the covering projection
$\pi^{\operatorname*{Mp}}:\operatorname*{Mp}(n)\longrightarrow
\operatorname*{Sp}(n)$ as $\pi^{\operatorname*{Mp}}(\widehat{V}_{P})=V_{P}$
and $\pi^{\operatorname*{Mp}}(\widehat{M}_{L,m})=M_{L}$. We use the
denomination \textquotedblleft local metaplectic group\textquotedblright%
\ because the products $\widehat{V}_{P}\widehat{M}_{L,m}$ are the only local
operators in $\operatorname*{Mp}(n)$: in harmonic analysis an operator is said
to be \textquotedblleft local\textquotedblright\ if it does not increase the
supports of the functions to which it is applied. For instance the modified
Fourier transform $\widehat{J}\in\operatorname*{Mp}(n)$ is not local since,
for example, $\widehat{J}\psi$ cannot be of compact support if $\psi\neq0$
because $\widehat{J}\psi$ is an analytic function in view of Paley--Wiener's
theorem. More generally, the Fourier integral operators (\ref{swm}) are never
local as can be seen by letting them act on a Dirac $\delta$ distribution.

We defined the inhomogeneous symplectic group $\operatorname*{ISp}(n)$ as
being the group generated by $\operatorname*{Sp}(n)$ and the translations
$T(z_{0}):z\longmapsto z_{0}$. Similarly, we define the inhomogeneous
metaplectic group $\operatorname*{IMp}(n)$ as the group of unitary operators
generated by $\operatorname*{Mp}(n)$ and the Heisenberg--Weyl operators
$\widehat{T}(z_{0})$, defined by \cite{Birk,Birkbis,Littlejohn}
\[
\widehat{T}(z_{0})\psi(x)=e^{\frac{i}{\hbar}(p_{0}x-\frac{1}{2}p_{0}x_{0}%
)}\psi(x-x_{0}).
\]
These operators satisfy the Weyl relations
\begin{align}
\widehat{T}(z_{0}+z_{1})  &  =e^{-\tfrac{i}{2\hslash}\sigma(z_{0},z_{1}%
)}\widehat{T}(z_{0})\widehat{T}(z_{1})\label{tzotzo}\\
\widehat{T}(z_{0})\widehat{T}(z_{1})  &  =e^{\tfrac{i}{\hslash}\sigma
(z_{0},z_{1})}\widehat{T}(z_{1})\widehat{T}(z_{0}). \label{tzotzobis}%
\end{align}
The inhomogeneous metaplectic group $\operatorname*{IMp}(n)$ consists of all
products
\begin{equation}
\widehat{S}\widehat{T}(z_{0})=\widehat{T}(Sz_{0})\widehat{S}. \label{dzo}%
\end{equation}
We will denote by $\operatorname*{IMp}_{0}(n)$ the subgroup of
$\operatorname*{IMp}(n)$ generated by $\operatorname*{Mp}_{0}(n)$ and the
Heisenberg--Weyl operators; it consists of all products $\widehat{S}%
\widehat{T}(z_{0})$ or $\widehat{T}(z_{0})\widehat{S}$ with $\widehat{S}%
\in\operatorname*{Mp}_{0}(n)$: we have%
\begin{equation}
\widehat{T}(z_{0})\widehat{V}_{P}\widehat{M}_{L,m}=\widehat{V}_{P}%
\widehat{M}_{L,m}\widehat{T}[M_{L^{-1}}V_{-P}z_{0}] \label{pr2}%
\end{equation}
which is the metaplectic version of (\ref{pr1}).

\subsection{Wigner transforms of Gaussians}

There is an immense literature about Gaussians and their Wigner transforms.
See for instance \cite{coro12,Littlejohn,sisumu}.

The Wigner transform of a function $\psi\in L^{2}(\mathbb{R}^{n})$ is the
function $W\psi\in L^{2}(\mathbb{R}^{2n})$ defined by%
\begin{equation}
W\psi(x,p)=\left(  \tfrac{1}{2\pi\hbar}\right)  ^{n}\int e^{-\frac{i}{\hbar
}py}\psi(x+\tfrac{1}{2}y)\psi^{\ast}(x-\tfrac{1}{2}y)d^{n}y. \label{wigtra}%
\end{equation}
This function is covariant under the action of the inhomogeneous groups
$\operatorname*{ISp}(n)$ and $\operatorname*{IMp}(n)$ in the sense that
\begin{equation}
W(\widehat{S}\psi)(z)=W\psi(S^{-1}z)\text{ \ , \ }W(\widehat{T}(z_{0}%
)\psi)(z)=W\psi(T(z_{0})z) \label{cov}%
\end{equation}
for all $\widehat{S}\in\operatorname*{Mp}(n)$ with projection $S\in
\operatorname*{Sp}(n)$ (see for instance \cite{Wigner,Littlejohn}). Of
particular interest to us are the Wigner transforms of non-degenerate Gaussian
functions (\textquotedblleft generalized squeezed coherent
states\textquotedblright)
\begin{equation}
\phi_{X,Y}(x)=(\pi\hbar)^{-n/4}(\det X)^{1/4}e^{-\frac{1}{2\hbar}(X+iY)x^{2}}
\label{fxy1}%
\end{equation}
($X$ and $Y$ real symmetric $n\times n$ matrices, $X>0$); one has the
well-known formula \cite{Bastiaans,Birk,Birkbis,Wigner,Littlejohn}
\begin{equation}
W\phi_{X,Y}(z)=(\pi\hbar)^{-n}e^{-\tfrac{1}{\hbar}Gz^{2}} \label{phagauss}%
\end{equation}
where $G$ is the symplectic symmetric positive definite matrix%
\begin{equation}
G=%
\begin{pmatrix}
X+YX^{-1}Y & YX^{-1}\\
X^{-1}Y & X^{-1}%
\end{pmatrix}
\in\operatorname*{Sp}(n). \label{gsym}%
\end{equation}
Noting that we can write $G=S_{X,Y}^{T}S_{X,Y}$ where%
\begin{equation}
S_{X,Y}=M_{X^{-1/2}}V_{-Y}=%
\begin{pmatrix}
X^{1/2} & 0\\
X^{-1/2}Y & X^{-1/2}%
\end{pmatrix}
\in\operatorname*{Sp}\nolimits_{0}(n) \label{bi}%
\end{equation}
we thus have
\begin{equation}
G=(M_{X^{-1/2}}V_{-Y})^{T}M_{X^{-1/2}}V_{-Y}. \label{G}%
\end{equation}
Notice that when $X=I_{\mathrm{d}}$ and $Y=0$ the Gaussian $\phi_{X,Y}$
reduces to the \textquotedblleft standard coherent state\textquotedblright%
\ \cite{Littlejohn}
\begin{equation}
\phi_{0}(x)=(\pi\hbar)^{-n/4}e^{-|x|^{2}/2\hbar} \label{pho1}%
\end{equation}
whose Wigner transform is simply%
\begin{equation}
W\phi_{0}(z)=(\pi\hbar)^{-n}e^{-|z|^{2}/\hbar}. \label{pho2}%
\end{equation}

More generally we will consider Gaussians centered at an arbitrary point; we
define the Gaussian $\phi_{X,Y}^{z_{0}}$ centered at $z_{0}$ as
\begin{equation}
\phi_{X,Y}^{z_{0}}=\widehat{T}(z_{0})\phi_{X,Y} \label{fizo}%
\end{equation}
where $\widehat{T}(z_{0})$ is the Heisenberg--Weyl operator, and we have,
using the translational covariance of the Wigner transform (second formula
(\ref{cov}))
\[
W\phi_{X,Y}^{z_{0}}(z)=(\pi\hbar)^{-n}e^{-\tfrac{1}{\hbar}G(z-z_{0})^{2}}.
\]
Setting $\Sigma^{-1}=\tfrac{2}{\hbar}G$ we can rewrite the Gaussian
(\ref{phagauss}) as
\[
W\phi_{X,Y}(z)=(2\pi)^{-n}\sqrt{\det\Sigma^{-1}}e^{-\frac{1}{2}\Sigma
^{-1}z^{2}}%
\]
which immediately leads to the following statistical interpretation: the
$2n\times2n$ matrix
\[
\Sigma=\frac{\hbar}{2}%
\begin{pmatrix}
X^{-1} & -X^{-1}Y\\
-YX^{-1} & X+YX^{-1}Y
\end{pmatrix}
\]
is the covariance matrix \cite{Littlejohn,optimal,Birk,stat} of the Gaussian
state $\phi_{X,Y}$.

The connection with the uncertainty principle is the following. We write the
covariance matrix in traditional form
\begin{equation}
\Sigma=%
\begin{pmatrix}
\Delta(x,x) & \Delta(x,p)\\
\Delta(p,x) & \Delta(p,p)
\end{pmatrix}
\label{covmat}%
\end{equation}
where $\Delta(x,x)=(\Delta(x_{j},x_{k}))_{1\leq j,k\leq n}$ and so on. Since
$G$ is symplectic and symmetric positive definite so is its inverse and the
$n\times n$ blocks $\Delta(x,x)=\Delta(x,x)^{T}$, $\Delta(x,p)=\Delta
(p,x)^{T}$, and $\Delta(p,p)=\Delta(p,p)^{T}$ must satisfy the relation%
\[
\Delta(x,x)\Delta(p,p)-\Delta(x,p)^{2}=\tfrac{1}{4}\hbar^{2}I_{\mathrm{d}}%
\]
as follows from the conditions (\ref{1cond}) or (\ref{2cond}) on the blocks of
a symplectic matrix. The latter implies that we must have%
\begin{equation}
(\Delta x_{j})^{2}(\Delta p_{j})^{2}=\Delta(x_{j},p_{j})^{2}+\tfrac{1}{4}%
\hbar^{2}\text{ , }1\leq j\leq n \label{RS}%
\end{equation}
which means that the so-called Robertson--Schr\"{o}dinger inequalities are
saturated (\textit{i.e.}become equalities). Notice that in particular we have
the textbook Heisenberg inequalities $\Delta x_{j}\Delta p_{j}\geq\frac{1}%
{2}\hbar$, which are a weaker form of the Robertson--Schr\"{o}dinger
uncertainty principle; they lack any symplectic invariance property, and
should therefore be avoided in any precise discussion of the uncertainty
principle (see the discussions in \cite{Birk,FOOP,goluPR}).

We will denote by $\operatorname*{Gauss}(n)$ the set of all Gaussian functions
(\ref{fizo}); we will identify that set with the set of all symplectic balls
with radius $R=\sqrt{\hbar}$. We will write $\operatorname*{Gauss}_{0}(n)$ for
the subset consisting of Gaussians centered at the origin.

\subsection{Quantum blobs}

We have introduced the notion of \textquotedblleft quantum
blob\textquotedblright\ in \cite{FOOP,blobs,goluPR}. A quantum blob
$Q_{S}^{2n}(z_{0})$ is a symplectic ball with radius $\sqrt{\hbar}$:
\begin{equation}
Q_{S}^{2n}(z_{0})=T(z_{0})SB^{2n}(\sqrt{\hbar}) \label{blob1}%
\end{equation}
and thus has symplectic capacity $\pi\hbar$. In view of Gromov's non-squeezing
theorem, every ellipsoid $\Omega$ in $\mathbb{R}^{2n}$ with symplectic
capacity $c(\Omega)\geq\pi\hbar$ contains a quantum blob. It is easy to see
that%
\begin{equation}%
{\textstyle\bigcap_{S\in\operatorname*{Sp}(n)}}
Q_{S}^{2n}(z_{0})=\{z_{0}\} \label{inter}%
\end{equation}
(it is sufficient to assume $z_{0}=0$ and $S=M_{\lambda I_{\mathrm{d}}}$ with
arbitrary $\lambda\neq0$).

Quantum blobs are minimum uncertainty phase space ellipsoids as follows from
the discussion in the previous subsection where we showed that Gaussians
saturate the Robertson--Schr\"{o}dinger principle (\ref{RS}):

\begin{proposition}
Let $\Sigma$ be the covariance matrix (\ref{covmat}) of a coherent state
(\ref{fizo}). The covariance ellipsoid
\begin{equation}
\Omega_{\Sigma}=\{z\in\mathbb{R}^{2n}:\tfrac{1}{2}\Sigma^{-1}z^{2}\leq1\}
\label{cosig}%
\end{equation}
is a quantum blob.
\end{proposition}

\begin{proof}
Since $\Sigma^{-1}=\tfrac{2}{\hbar}G$, the covariance ellipsoid $\Sigma$ is
equivalently determined by the inequality $Gz^{2}\leq\hbar$; in view of the
factorization (\ref{G}) of $G$ we thus have
\begin{equation}
\Omega_{\Sigma}=V_{Y}M_{X^{1/2}}B^{2n}(\sqrt{\hbar})=Q_{V_{Y}M_{X^{1/2}}}%
^{2n}(0). \label{cosigbis}%
\end{equation}

\end{proof}

The ellipsoid (\ref{cosig}) is called the \textit{Wigner ellipsoid} in some texts.

All quantum blobs can be built from the elementary quantum blob $B^{2n}%
(\sqrt{\hbar})$, which is the covariance ellipsoid of the standard coherent
state (\ref{pho1}):

Every quantum blob $Q_{S}^{2n}(z_{0})$ can be generated from the ball
$B^{2n}(\sqrt{\hbar})$ using the local subgroup $\operatorname*{ISp}_{0}(n)$
of $\operatorname*{ISp}(n)$. In fact (Proposition \ref{propsplp}), for every
$S\in\operatorname*{Sp}(n)$ there exist unique $P=P^{T}$, $L=L^{T}$, and
$z_{0}\in\mathbb{R}^{2n}$ such that \
\begin{equation}
Q_{S}^{2n}(z_{0})=T(z_{0})V_{P}M_{L}B^{2n}(\sqrt{\hbar}). \label{blobimp}%
\end{equation}
More generally, it immediately follows from Proposition \ref{propsplp} that:

\begin{proposition}
The group $\operatorname*{Sp}_{0}(n)$ acts transitively on
$\operatorname*{Quant}_{0}(2n)$ and $\operatorname*{ISp}\nolimits_{0}(n)$ acts
transitively on $\operatorname*{Quant}(2n)$. Explicitly, if $S=V_{P}M_{L}$ and
$S^{\prime}=V_{P^{\prime}}M_{L^{\prime}}$ then%
\begin{equation}
Q_{S^{\prime}}^{2n}(z_{0}^{\prime})=S(P,L,P^{\prime},L^{\prime},z_{0}%
,z_{0}^{\prime})Q_{S}^{2n}(z_{0}) \label{qprimeq}%
\end{equation}
with $S(P,L,P^{\prime},L^{\prime},z_{0},z_{0}^{\prime})\in\operatorname*{ISp}%
_{0}(n)$ being given by (\ref{A10}), (\ref{B10}), and (\ref{C11}).
\end{proposition}

We will denote by $\operatorname*{Quant}(2n)$ the set of all quantum blobs in
$\mathbb{R}^{2n}$; the subset of $\operatorname*{Quant}(2n)$ consisting of
quantum blobs $Q_{S}^{2n}(0)$ centered at $0$ will be denoted
$\operatorname*{Quant}_{0}(2n)$. The set $\operatorname*{Quant}(2n)$ plays the
role of a \textit{quantum phase space}. It will be equipped with the topology
induced by the Hausdorff distance (\ref{Haus1}).

\subsection{The correspondence between quantum blobs and Gaussians}

Consider again the standard coherent state (\ref{pho1}):
\[
\phi_{0}(x)=(\pi\hbar)^{-n/4}e^{-|x|^{2}/2\hbar}%
\]
and let $\widehat{S}\in\operatorname*{Mp}(n)$ be a metaplectic operator with
projection
\[
S=%
\begin{pmatrix}
A & B\\
C & D
\end{pmatrix}
\]
on $\operatorname*{Sp}(n)$. One can calculate $\widehat{S}\phi_{0}$ as follows
\cite{coro97,Littlejohn}: one first assumes that $S$ is a free symplectic
matrix (i.e. $\det B\neq0$) so that $\widehat{S}$ is a Fourier integral
operator (\ref{swm}); a tedious but straightforward calculation of Gaussian
integrals then yields the explicit formula
\begin{equation}
\widehat{S}\phi_{0}(x)=(\pi\hbar)^{-n/4}K\exp\left(  \frac{i}{2\hbar}\Gamma
x^{2}\right)  \label{fimeta}%
\end{equation}
where $K$ and $\Gamma$ are defined by
\begin{equation}
K=(\det(A+iB))^{-1/2}\text{ and }\Gamma=(C+iD)(A+iB)^{-1}; \label{fimetabis}%
\end{equation}
the argument of $\det(A+iB)\neq0$ depends on the choice of the operator
$\widehat{S}\in\operatorname*{Mp}(n)$ with projection $S\in\operatorname*{Sp}%
(n)$ (there are several ways of proving that $A+iB$ is invertible and that
$\Gamma$ is symmetric; see for instance \cite{coro06,Birkbis,Littlejohn}).
Now, this can be considerably simplified if one uses local metaplectic
operators. We begin by remarking that if $S=U\in U(n)$ then we have%
\[
U=%
\begin{pmatrix}
X & Y\\
-Y & X
\end{pmatrix}
,
\]
the blocks $X$ and $Y$ satisfying the conditions (\ref{xxyy1}), (\ref{xxyy2}).
It follows that $(C+iD)(A+iB)^{-1}=-i$ and, since $X+iY\in U(n,\mathbb{C})$
that $|\det(A+iB)|=|\det(X+iY)|=1$. Formulas (\ref{fimeta})--(\ref{fimetabis})
thus lead to
\begin{equation}
\widehat{U}\phi_{0}=i^{\gamma}\phi_{0} \label{ufo}%
\end{equation}
where $\gamma$ is a real phase associated to a choice of argument of
$\det(X+iY)$. The standard Gaussian $\phi_{0}$ is thus an eigenfunction of
every metaplectic operator arising from a symplectic rotation. This
observation allows a considerable simplification in the derivation of formula
(\ref{fimeta}):

\begin{proposition}
Let $\widehat{S}\in\operatorname*{Mp}(n)$ have projection $S=V_{-P}M_{L}U$ on
$\operatorname*{Sp}(n)$ (pre-Iwasawa factorization). Then%
\begin{equation}
\widehat{S}\phi_{0}=i^{\gamma}\widehat{V}_{-P}\widehat{M}_{L,m}\phi_{0}
\label{sifo1}%
\end{equation}
with $m\in\{0,2\}$, that is, explicitly,
\begin{equation}
\widehat{S}\phi_{0}(x)=\frac{i^{m+\gamma}}{(\pi\hbar)^{n/4}}\sqrt{\det
L}e^{-\frac{1}{2\hbar}(iP-L^{2})x^{2}}\exp\left(  -\frac{1}{2\hbar}%
(iP-L^{2})x^{2}\right)  \label{sifo2}%
\end{equation}
with $P=P^{T}$, $L=L^{T}>0$ being given by
\begin{align}
P  &  =(CA^{T}+DB^{T}-I_{\mathrm{d}})(AA^{T}+BB^{T})^{-1}\label{P}\\
L  &  =(AA^{T}+BB^{T})^{-1/2}. \label{L}%
\end{align}

\end{proposition}

\begin{proof}
We have $\widehat{S}=\widehat{V}_{-P}\widehat{M}_{L,m}\widehat{U}$ with
$m\in\{0,2\}$ (because $\det L>0$). In view of formula (\ref{ufo}) we have%
\[
\widehat{S}\phi_{0}=i^{\gamma}\widehat{V}_{-P}\widehat{M}_{L,m}\phi_{0}%
\]
which is (\ref{sifo1}). Formula (\ref{sifo2}) follows using the expressions
(\ref{pl1}) and (\ref{pl2}) for the matrices $P$ and $L$.
\end{proof}

It requires a modest number of matrix calculations to verify that
(\ref{sifo2}) is equivalent to (\ref{fimeta}). The argument goes as follows:
one rewrites $\Gamma$ in (\ref{fimeta}) as
\[
\Gamma=(C+iD)(A^{T}-iB^{T})\left[  (A+iB)(A^{T}-iB^{T})\right]  ^{-1}%
\]
and one expands the products taking into account the relations (\ref{1cond}%
)--(\ref{2cond}) satisfied by the matrices $A,B,C,D$, which leads to
\[
\Gamma=(C+iD)(A^{T}-iB^{T})\left[  (A+iB)(A^{T}-iB^{T})\right]  ^{-1}.
\]
We leave the computational details to the reader.

The result above allows us to prove that there is a natural bijection%
\[
\operatorname*{Quant}(2n)\overset{\approx}{\longleftrightarrow}%
\operatorname*{Gauss}(n)
\]
allowing to construct a commutative diagram%
\[%
\begin{array}
[c]{ccc}%
\operatorname*{IMp}\nolimits_{0}(n)\times\operatorname*{Gauss}(n) &
\longrightarrow & \operatorname*{Gauss}(n)\\
\downarrow &  & \downarrow\\
\operatorname*{ISp}\nolimits_{0}(n)\times\operatorname*{Quant}(2n) &
\longrightarrow & \operatorname*{Quant}(2n).
\end{array}
\]
Let us study these properties in detail. We begin by proving the
correspondence between quantum blobs and Gaussians.

\begin{proposition}
The natural mapping%
\begin{equation}
Q_{S}^{2n}(z_{0})\longmapsto\phi_{X,Y}^{z_{0}}=\widehat{T}(z_{0}%
)\widehat{S}\phi_{0}, \label{Q2nfi}%
\end{equation}
where $T(z_{0})S$ is the projection of $\widehat{T}(z_{0})\widehat{S}%
\in\operatorname*{IMp}_{0}(n)$ on $\in\operatorname*{ISp}_{0}(n)$ is a
bijection
\begin{equation}
\operatorname*{Quant}(2n)\longrightarrow\operatorname*{Gauss}(n).
\label{bijection}%
\end{equation}

\end{proposition}

\begin{proof}
Since $T(z_{0})S$ is uniquely determined the mapping (\ref{bijection}) is
well-defined. To prove that it it is a bijection it suffices to note that the
relation $\phi_{X,Y}^{z_{0}}=\widehat{T}(z_{0})\widehat{S}\phi_{0}$
unambiguously determines $\widehat{T}(z_{0})\widehat{S}$ and hence also
$T(z_{0})S$.
\end{proof}

The following statement is the quantum analogue of part \textit{(ii)}
of\textit{ }Proposition \ref{propsplp}.

\begin{proposition}
\label{Thm3}The local inhomogeneous metaplectic group $\operatorname*{IMp}%
_{0}(n)$ acts transitively on the Gaussian phase space $\operatorname*{Gauss}%
(n)$. In fact, for any two Gaussians $\phi_{X,Y}^{z_{0}}$ and $\phi
_{X^{\prime},Y^{\prime}}^{z_{0}^{\prime}}$ we have%
\begin{equation}
\phi_{X^{\prime},Y^{\prime}}^{z_{0}^{\prime}}=e^{\tfrac{i}{\hbar}\chi
(z_{0},z_{0}^{\prime})}\widehat{T}(z_{0}^{\prime\prime})\widehat{V}%
_{P^{\prime\prime}}\widehat{M}_{L^{\prime\prime},0}\phi_{X,Y}^{z_{0}}
\label{transit1}%
\end{equation}
with%
\begin{align}
\chi(z_{0},z_{0}^{\prime})  &  =\tfrac{1}{2}\sigma(z_{0}^{\prime}%
,-Rz_{0})\label{transit2}\\
L^{\prime\prime}  &  =X^{-1/2}X^{\prime1/2}\text{ \ , \ }P^{\prime\prime
}=Y^{\prime}-L^{\prime\prime}Y(L^{\prime\prime})^{T}\label{transit3}\\
z_{0}^{\prime\prime}  &  =z_{0}^{\prime}-V_{P^{\prime\prime}}M_{L^{\prime
\prime}}z_{0}. \label{transit4}%
\end{align}

\end{proposition}

\begin{proof}
We have $\phi_{X,Y}^{z_{0}}=\widehat{T}(z_{0})\widehat{S}\phi_{0}$ and
$\phi_{X^{\prime},Y^{\prime}}^{z_{0}^{\prime}}=\widehat{T}(z_{0}^{\prime
})\widehat{S}^{\prime}\phi_{0}$ with $\widehat{S}=\widehat{V}_{Y}%
\widehat{M}_{X^{1/2},0}$ and $\widehat{S}^{\prime}=\widehat{V}_{Y^{\prime}%
}\widehat{M}_{X^{\prime1/2},0}$; hence
\[
\phi_{X^{\prime},Y^{\prime}}^{z_{0}^{\prime}}=\widehat{T}(z_{0}^{\prime
})\widehat{S}^{\prime}(\widehat{T}(z_{0})\widehat{S})^{-1}\phi_{X,Y}^{z_{0}%
}=\widehat{T}(z_{0}^{\prime})\widehat{S}^{\prime}\widehat{S}^{-1}%
\widehat{T}(-z_{0})\phi_{X,Y}^{z_{0}}.
\]
Using successively formulas (\ref{dzo}) and (\ref{tzotzo}) we have%
\[
\widehat{T}(z_{0}^{\prime})\widehat{S}^{\prime}\widehat{S}^{-1}\widehat{T}%
(-z_{0})=e^{\tfrac{i}{2\hslash}\sigma(z_{0}^{\prime},-S^{\prime}S^{-1}z_{0}%
)}\widehat{T}(z_{0}^{\prime}-S^{\prime}S^{-1}z_{0})\widehat{S}^{\prime
}\widehat{S}^{-1}.
\]
Using formulas (\ref{xinvy}) and (\ref{XINVY}) with $P=Y$, $P^{\prime
}=Y^{\prime}$, $L=X^{1/2}$, and $L^{\prime}=X^{\prime1/2}$, we get
\[
\widehat{S^{\prime}}\widehat{S}^{-1}=\widehat{V}_{Y^{\prime}-(X^{-1/2}%
X^{\prime1/2})^{T}Y(X^{-1/2}X^{\prime1/2})}\widehat{M}_{X^{-1/2}X^{\prime
1/2},0}%
\]
and the projection of $\widehat{S^{\prime}}\widehat{S}^{-1}$ on
$\operatorname*{ISp}\nolimits_{0}(n)$ is
\[
S^{\prime}S^{-1}=V_{Y^{\prime}-(X^{-1/2}X^{\prime1/2})^{T}Y(X^{-1/2}%
X^{\prime1/2})}M_{X^{-1/2}X^{\prime1/2}},
\]
hence formulas (\ref{transit2}) and (\ref{transit3}).
\end{proof}

\section{Discussion and Perspectives}

The theory of \textquotedblleft chalkboard motion\textquotedblright\ we have
outlined in this paper might have applications to several important topics in
mathematics and mathematical physics. In particular:

\begin{itemize}
\item \textit{Celestial mechanics}: recent work of Scheeres and collaborators
shown the important role played by techniques from symplectic topology in
guidance and control theory (see for instance \cite{sch} which uses symplectic
capacities to study spacecraft trajectory uncertainty). The approach outlined
in the present work could certainly be used with success in analyzing
planetary motions since we do not have to solve directly complicated Hamilton
equations arising in, say, the many-body problem, but rather control the
trajectories at every step. Numerical algorithms such as symplectic
integrators could certainly be easily implemented here following the work of
Feng and Qin \cite{feng} or Wang \cite{Wang};

\item \textit{Entropy}: in \cite{kalo2,kalo1} Kalogeropoulos applies
techniques for symplectic topology (the non-squeezing theorem) to the study of
various notions of entropy in the context of thermodynamics; this approach
seems to be very promising; we will return in a future work to the
applications of chalkboard motion to these important questions where
coarse-graining methods have played historically an important role;

\item \textit{Semiclassical methods}: the nearby orbit method we have used to
study chalkboards motions both from a classical and quantum perspective
originate from robust techniques which have been used for decades in
semiclassical mechanics to approximate non-linear motions;

\item \textit{Collective motions}: An interesting property of Hamiltonian
symplectomorphisms due to Boothby \cite{Boothby} is \textquotedblleft$N$-fold
transitivity\textquotedblright: given two arbitrary sets $\{z_{1},...,z_{N}\}$
and $\{z_{1}^{\prime},...,z_{N}^{\prime}\}$ of $N$ distinct points in
$\mathbb{R}^{2n}$, Boothby proved that there exists $f\in\operatorname*{Ham}%
(n)$ such that $z_{j}^{\prime}=f(z_{j})$ for every $j\in\{1,...,N\}$. A
natural question would then be \textquotedblleft given $N$ disjoint symplectic
balls $B_{S}^{2n}(z_{1},\varepsilon),...,B_{S}^{2n}(z_{N},\varepsilon)$ at
time $t=0$ can we find a Hamiltonian flow taking these balls to a new
configuration of disjoint symplectic balls $B_{S^{\prime}}^{2n}(z_{1}^{\prime
},\varepsilon),...,B_{S^{\prime}}^{2n}(z_{N}^{\prime},\varepsilon)$ at some
later time $t$? The difficulty here comes from the fact that in the course of
this collective motion the ellipsoids might \textquotedblleft
collide\textquotedblright\ and have a non-empty intersection: conservation of
symplectic capacity (and even volume) has nothing to do with conservation of
shape, so two adjacent initially non-intersecting ellipsoids might very well
intersect after a while, being stretched and sheared. So an answer to this
question might require strong limitations on the type of chalkboard motion we
can choose;

\item \textit{Study of subsystems of a Hamiltonian system}: Proposition
\ref{PropAB} says that the orthogonal projection of a symplectic phase space
ball on a phase space with a smaller dimension also contains a symplectic ball
with the same radius. In the quantum case, these symplectic balls are just
quantum blobs. This projection result is the key to the applications of the
theory of chalkboard motion to subsystems. We are currently investigating the
topic; for some preliminary results see our preprint \cite{goAB}.
\end{itemize}

\begin{acknowledgement}
This work has been supported by the Grants P23902-N13 and P27773-N25 of the
Austrian Research Foundation FWF.
\end{acknowledgement}

\begin{acknowledgement}
The author would like to thank Glen Dennis and Basil Hiley (UCL), and Leonid
Polterovich (Tel Aviv) for constructive criticism and encouragement.
\end{acknowledgement}

\end{document}